\documentclass[11pt]{article}

\usepackage[a4paper,margin=1in]{geometry}
\usepackage[T1]{fontenc}
\usepackage[utf8]{inputenc}
\usepackage{lmodern}
\usepackage{microtype}
\usepackage{amsmath,amssymb,amsthm,mathtools}
\usepackage{booktabs}
\usepackage{enumitem}
\usepackage{xcolor}
\usepackage{tikz}
\usepackage{aliascnt}
\usepackage{hyperref}
\usepackage[nameinlink,capitalise,noabbrev]{cleveref}

\hypersetup{
  colorlinks=true,
  linkcolor=blue!55!black,
  citecolor=green!40!black,
  urlcolor=blue!60!black
}

\newtheorem{theorem}{Theorem}[section]
\newaliascnt{lemma}{theorem}
\newtheorem{lemma}[lemma]{Lemma}
\aliascntresetthe{lemma}
\newaliascnt{corollary}{theorem}
\newtheorem{corollary}[corollary]{Corollary}
\aliascntresetthe{corollary}
\newaliascnt{proposition}{theorem}
\newtheorem{proposition}[proposition]{Proposition}
\aliascntresetthe{proposition}
\newaliascnt{observation}{theorem}
\newtheorem{observation}[observation]{Observation}
\aliascntresetthe{observation}
\theoremstyle{definition}
\newaliascnt{definition}{theorem}
\newtheorem{definition}[definition]{Definition}
\aliascntresetthe{definition}
\newaliascnt{example}{theorem}
\newtheorem{example}[example]{Example}
\aliascntresetthe{example}
\theoremstyle{remark}
\newaliascnt{remark}{theorem}
\newtheorem{remark}[remark]{Remark}
\aliascntresetthe{remark}

\newcommand{\HH}{\mathsf{H}}
\newcommand{\VV}{\mathsf{V}}
\newcommand{\dist}{\operatorname{dist}}
\newcommand{\corank}{\operatorname{corank}}
\newcommand{\R}{\mathbb{R}}
\newcommand{\Z}{\mathbb{Z}}
\newcommand{\Qcell}[2]{Q_{#1,#2}}
\newcommand{\GammaP}{\Gamma_{p}}
\newcommand{\Reeb}{R_{f_p}}
\newcommand{\calR}{\mathcal{R}_{p}}
\newcommand{\betti}{\beta_1}

\title{\textbf{The Reeb Structure of Bend Distance in Grid Domains:}\\
Cycle Bounds with Holes, Exact Sector Geometry on Disks, and the
Two-Port Constant \(c_2^\square=3\)}
\author{Aoji Li\thanks{Corresponding author.
Email: \texttt{aojili77@gmail.com}}
\and Guangbo Ding}
\date{August 3, 2026}

\begin{document}
\maketitle

\begin{abstract}
The fixed-parameter algorithm for coordinated motion planning on
discretized simple polygons (ICALP 2026) rests on a single-port
\emph{sector decomposition}: grid vertices are labelled by the minimum
number of bends needed to reach an oriented terminal.  We develop the
structure theory of this decomposition on grid domains with holes.  Once
its endpoint convention is made precise, bend distance extends to a
canonical piecewise-affine function on the domain, and the sector graph
is the parallel-edge shadow of the Reeb multigraph of that function.
For a finite pure planar cubical domain with \(h\) holes this yields
\(\betti(\Gamma_p)\le\betti(R_{f_p})\le h\), together with a linear-time
computable feedback set of at most \(h\) sector vertices;
coordinate-level one-hole examples show that the finer hole-free
geometry---unique predecessors and straight baselines---fails.  On
hole-free cubical disks the machinery is exact: every positive sector is
the one-sided extrusion of a unique straight parent interface, and the
two-port common refinement has treewidth exactly \(c_2^\square=3\),
although sector count, cycle rank, and feedback number are unbounded
already there.  The endpoint precision is necessary rather than
cosmetic: under the literal reading of the source's definition, one
hole-free unit square has sector graph \(C_3\), falsifying the
single-port tree lemma; the augmented convention used here is exactly
the metric computed by the source's own layering procedure.  We do
\emph{not} obtain an \(f(k,h)n^{O(1)}\) algorithm; the paper supplies
the structural first step and a falsification tool for that program.
\end{abstract}

\section{Introduction}
\label{sec:intro}

Coordinated motion planning (CMP) asks for collision-free routes for a
small number of labelled robots in a graph.  Under the total-travel
objective, an edge traversal has unit cost and waiting has zero cost.  The
problem is fixed-parameter tractable on full rectangular grids
\cite{eiben2023parameterized}, and more generally when parameterized by
the number of robots plus the treewidth of the input graph
\cite{deligkas2024energy}.  Deligkas, Eiben, Ganian, and Kanj recently
proved fixed-parameter tractability, parameterized only by the number of
robots, for grid graphs obtained from simple polygons without holes
\cite{deligkas2026polygon}.

A central object in the latter proof is a sector decomposition.  Fix an
oriented terminal, called a \emph{port}.  Every grid vertex is labelled by
the minimum number of axis changes---bends---needed to reach the port with
the prescribed final axis.  A \emph{sector} is a connected component of an
equal label, and the \emph{sector graph} records adjacency between
components.  In a hole-free discretized polygon, the single-port sector
graph is a tree.  Moreover, every non-root sector has a unique predecessor
and is generated from a separating straight baseline.  These geometric
facts feed a multi-port refinement and ultimately a bounded-treewidth
reduction \cite[full version, Lemmas~9--24]{deligkas2026polygon}.

The same paper explicitly leaves constant or parameter-bounded numbers of
holes as an open direction, and its authors expect that direction to be
attainable by refining their approach.  The first obstacle is visible
already for one port: a hole permits routes of the same bend count to
approach a sector from different sides.  The purpose of this paper is to
state exactly what remains true with holes and, equally importantly, what
does not.

\paragraph{Main idea.}
Once the endpoint convention is repaired, bend distance extends to a
canonical piecewise-affine function on the domain
(\cref{thm:extension}), and the sector decomposition is exactly the
level-set structure of that function: sectors are its integer contours,
and the sector graph is the parallel-edge shadow of its Reeb multigraph
(\cref{thm:shadow}).  This lens splits the hole-free sector theory of
\cite{deligkas2026polygon} into two parts of different nature.  The
\emph{topological} part survives holes in graded form: cycles of a
single-port sector graph can come only from holes, so
\(\betti(\GammaP)\le h\).  The \emph{geometric} part---unique
predecessors, straight baselines, bounded common refinements---is a disk
phenomenon: a single hole destroys each of these statements, while on
disks they are not merely true but exact, down to the precise two-port
constant \(c_2^\square=3\).  Every proof in the paper runs through one
technical device, the two-orientation \(0/1\)-weighted state graph of
\cref{sec:state}, whose parity structure yields the affine extension,
the preferred-axis geometry on disks, and the exact arithmetic of the
counterexample families.  The following table summarizes the fate of
each hole-free statement.

\begin{center}
\small
\begin{tabular}{@{}p{0.37\linewidth}p{0.43\linewidth}l@{}}
\toprule
hole-free statement in \cite{deligkas2026polygon} &
fate on domains with \(h\ge1\) holes & where \\
\midrule
formal bend-distance definition &
literal reading falsifies the tree lemma on one square; repaired by the
augmented stub convention & \cref{cor:literal-c3} \\
single-port sector graph is a tree &
becomes \(\betti(\GammaP)\le h\), tight, with a linear-time feedback
set & \cref{thm:holebound} \\
sector graph determines the level-set topology &
fails: parallel Reeb bands are hidden & \cref{ex:parallel} \\
unique predecessor &
fails: a one-hole \(C_4\) & \cref{ex:c4} \\
straight one-sided baselines &
fail even for a tree quotient & \cref{ex:nobaseline} \\
bounded two-port refinement &
treewidth exactly \(3\) on filled disks; open with holes &
\cref{cor:sharp-two-port-constant} \\
\bottomrule
\end{tabular}
\end{center}

\paragraph{Contributions.}
Our first main result is a topological invariant for single-port sectors
on a holed grid domain, stated for the \emph{augmented} port metric of
\cref{sec:prelim}---an endpoint-precise version of the source's bend
distance; the need for this precision is itself one of our findings and
is summarized at the end of this list.
For a finite connected pure planar cubical complex
\(X\subset\R^2\) with \(h\) holes and \(G=X^{(1)}\), we derive augmented
bend distance from an orientation-state shortest-path problem
(\cref{sec:state}).  A local parity argument shows that the four values on
every unit square are either constant or form a single horizontal or
vertical stripe (\cref{lem:nosaddle}); hence there is a canonical affine
extension \(f_p\) on every square, without choosing a triangulation
(\cref{thm:extension}).  We identify its integer contours with sectors and
its open inter-integer bands with the edges of a Reeb multigraph, and we
prove that the Reeb quotient of \(f_p\) is homeomorphic to this multigraph
while the simple sector graph is its parallel-edge simplification
(\cref{thm:shadow}).  Combining this identification with Gelbukh's
cycle-rank theorem for Reeb graphs \cite{gelbukh2019approximation} and
Alexander duality gives
\[
  \betti(\GammaP)\leq\betti(\Reeb)\leq\corank\pi_1(X)\leq h,
\]
together with a linear-time computable feedback set of at most \(h\)
sector vertices (\cref{thm:holebound,cor:algorithm}).

Our negative results separate three assertions that can otherwise be
conflated:
\begin{enumerate}[label=(\roman*),leftmargin=2.2em]
  \item the sector graph is the Reeb graph;
  \item the sector graph is a tree with unique predecessors; and
  \item a tree sector graph forces straight separator baselines.
\end{enumerate}
All three implications fail in the presence of one hole
(\cref{sec:examples}).  The examples are coordinate-level constructions,
so their labels and adjacencies can be checked without a drawing or a
generic-position assumption.

On filled cubical disks the cycle bound becomes an exact structure
theorem.  The Reeb multigraph itself is a tree, not merely a tree after
parallel bands are merged; preferred-axis parity therefore turns each
parent band into a unique straight strip and gives an exact one-sided
extrusion theorem: every positive sector is the one-sided extrusion of a
straight parent interface, which gates its descendant side
(\cref{thm:hf-osp}).  The resulting sectors are cubically convex.

Next, we show that the single-port cycle bound cannot be iterated
through the common refinement of several ports by bounding its size or
cycle rank.  A comb and an alternating tab construction give, already for
one port on a hole-free disk, unbounded degree-two-suppressed size and
unboundedly many branch vertices (\cref{thm:comb-tabs}).  A two-port
staircase family has linearly growing cycle rank and feedback number but
treewidth exactly two (\cref{thm:staircase,cor:stair-fvs}).  Thus
feedback-set and bounded-core-size arguments fail even though a
bounded-treewidth theorem can still hold; among the specific size,
cycle-rank, and feedback-set targets ruled out here, only treewidth (or
an equivalent adhesion-bounded target) remains viable.

With treewidth thus identified as the target, straight parent
interfaces, three-vertex adhesions, planar local completions, and a
clique-compatible gluing along the one-port tree
bound the two-port common refinement: \(\operatorname{tw}(\Gamma_{p,q})\le3\)
(\cref{thm:sharp-two-port-upper}).  The \(2\times2\) cubical disk attains
equality with \(\Gamma_{p,q}=P_3\square P_3\)
(\cref{prop:two-port-grid}), so the filled two-port constant is exactly
\(c_2^\square=3\) (\cref{cor:sharp-two-port-constant}).  This sharp equality is not
silently transferred to the broader edge-discretized source class of
\cite{deligkas2026polygon}; the concluding discussion
(\cref{sec:edge-discussion}) isolates a bridge/core lemma which, from
two explicitly stated hypotheses imported from the source, gives the
conditional window \(3\le c_2^{\mathrm{edge}}\le11\) and records a
misalignment in the public v1 arithmetic of the source's Theorem~16
(\cref{cor:two-port-eleven,rem:two-port-transfer,rem:two-port-source-arithmetic}).

Finally, an enabling correction.  The augmented metric is not a cosmetic
choice.  Under the literal last-actual-edge reading of ``arrival axis''
in \cite{deligkas2026polygon}, one unit square has bend values
\((0,0,1,2)\) and its sector graph is a triangle on a hole-free square,
so the single-port tree lemma and the cycle bound fail outright for that
reading (\cref{prop:endpoint,cor:literal-c3} and
\cref{fig:endpoint}).  We prove that the augmented metric is exactly the
one computed by the operational layering procedure in the CMP proof
(\cref{prop:layering}); our positive results do not assert an affine or
Reeb theorem for the literal metric.

\paragraph{What this paper does not do.}
The cycle bound is not a solution to CMP with holes.  It controls graph
cycles in each \emph{single-port} quotient, but not straight baselines,
metric-preserving reductions, multi-port common refinements, collision
lifting, or the treewidth of a reduced input graph.
\Cref{sec:conclusion} makes the remaining gap explicit.  We do \emph{not}
obtain an \(f(k,h)n^{O(1)}\) algorithm; that extension remains open.

\paragraph{Related work.}
\emph{Multi-robot motion planning on graphs.}
Feasibility of coordinated motion on graphs is classical
\cite{kornhauser1984pebble,papadimitriou1994motion}, while optimization
versions are hard already in severely restricted settings: the
\((n^2-1)\)-puzzle \cite{ratner1990puzzle}, distance- and
makespan-optimal multi-agent path finding
\cite{yu2013structure,geft2022refined}, grid graphs with holes
\cite{banfi2017intractability}, and unlabeled variants
\cite{calinescu2008reconfigurations,solovey2016hardness}.  Algorithmic
counterpoints include constant-stretch reconfiguration of dense swarms
\cite{demaine2019coordinated} and the CG:SHOP 2021 challenge
\cite{fekete2022cgshop}.  The parameterized line most relevant here
comprises fixed-parameter algorithms in the number of robots on full
grids \cite{eiben2023parameterized}, in the number of robots plus the
treewidth of the input graph \cite{deligkas2024energy}, for sliding
robots via minor testing \cite{eiben2025sliding}, and in the number of
robots on discretized simple polygons \cite{deligkas2026polygon}, whose
sector machinery this paper studies; sectors themselves originate in a
drawing-extension setting \cite{bhore2023extending}.  In the continuous
setting, an exact polynomial-time algorithm for the min-sum objective
was recently obtained for two square robots in a rectilinear
environment, possibly with holes, while the min-makespan variant is
NP-hard \cite{agarwal2025twosquare}.

\emph{Link distance.}
Bend distance is the grid-graph analogue of rectilinear link distance.
Minimum-link paths admit linear-time algorithms in simple polygons
\cite{suri1986linear}; rectilinear link distance was studied by de~Berg
\cite{deberg1991rectilinear}, minimum-link paths among obstacles in
\cite{mitchell1992minimum,mitchell2015minimumlink,wang2017bicriteria},
and the area is surveyed in \cite{maheshwari2000link}.  That literature
computes distances, diameters, and paths; the present paper instead
studies the topology of the level sets of a link-type metric.

\emph{Reeb graphs.}
Reeb graphs originate in Morse theory \cite{reeb1946points} and are a
standard tool in shape analysis \cite{biasotti2008reeb}.  Loops of Reeb
graphs of \(2\)-manifolds were analyzed in
\cite{colemclaughlin2004loops}; the corank bound used here is due to
Gelbukh \cite{gelbukh2019approximation,gelbukh2018loops}.  Efficient
construction and approximation of Reeb graphs are treated in
\cite{parsa2013deterministic,dey2013reeb}, and categorical aspects in
\cite{desilva2016categorified}.  Our contribution on this side is a
discrete-to-PL realization theorem, proved directly for a metric on a
grid graph rather than via Morse-theoretic genericity.

\emph{CAT(0) cube complexes and treewidth.}
Hyperplanes and convex halfspaces of cube complexes go back to Sageev
\cite{sageev1995ends} (see \cite{farley2009sageev} for the statement
used here); the one-skeleta of CAT(0) cube complexes are exactly the
median graphs \cite{chepoi2000graphs}.  In robotics these complexes
usually arise as \emph{configuration} spaces---state complexes and
geodesic motion planning \cite{abrams2004state,ardila2014moving}---%
whereas in this paper the \emph{domain} itself is the CAT(0) object, and
hyperplane convexity of extrusions drives the two-port treewidth bound.
For treewidth background see
\cite{robertson1986graph,bodlaender1998partial}.

\paragraph{Scope.}
Every positive bend-distance and Reeb-graph theorem below concerns the
\emph{augmented} port metric on the one-skeleton of a finite, compact,
connected, pure two-dimensional cubical subcomplex of \(\R^2\).  The pure
cubical model is the natural one obtained by filling free unit cells.
Some polygon-discretization conventions instead retain every grid vertex
and edge contained in a polygonal free space; such a graph need not be the
one-skeleton of a pure unit-square complex.  We flag explicitly when an
example is in this broader edge-discretization model.  The distinction is
essential: the canonical squarewise extension requires every edge and
vertex to be supported by full squares.

\paragraph{Organization.}
\Cref{sec:prelim} repairs the metric and proves that the literal reading
falsifies the source's tree lemma.
\Cref{sec:state,sec:extension,sec:reeb,sec:bound} form a single pipeline
from the orientation-state graph to the cycle bound
\(\betti(\GammaP)\le h\).  \Cref{sec:examples} gives the three one-hole
obstructions that delimit this bound.  \Cref{sec:hf-single-port} then
reverses the lens and proves the exact one-sided geometry on cubical
disks.  \Cref{sec:joint} treats two-port common refinements in the same
order as the story: the wrong-parameter families and the sharp
filled-disk constant \(c_2^\square=3\).
\Cref{sec:edge-discussion} discusses the broader edge-discretized
scope, where our two-port bound is conditional, and
\cref{sec:conclusion} collects what the cycle bound does not prove and
the open problems.

\section{Cubical domains, bend distance, and sectors}
\label{sec:prelim}

This section fixes the object of study.  We define the literal and the
augmented port metrics, show on a single unit square that the literal
reading falsifies the tree lemma of \cite{deligkas2026polygon}, and
prove that the augmented metric is the one computed by the operational
layering procedure of that work.

For \((x,y)\in\Z^2\), write
\[
  \Qcell{x}{y}=[x,x+1]\times[y,y+1].
\]
A \emph{finite planar cubical domain} is a finite cubical complex
\(X\subset\R^2\) obtained as a union of unit squares and all their faces.
Throughout the positive part of the paper, \(X\) is connected and
\emph{pure two-dimensional}: every edge and vertex of \(X\) lies in at least
one unit square of \(X\).  Set \(G=X^{(1)}\).  The graph \(G\) is finite and
connected.

The number of holes is
\[
  h=\#\{\text{bounded connected components of }\R^2\setminus X\}.
\]
Thus \(S^2\setminus X\) has \(h+1\) connected components.

\begin{definition}[Literal and augmented port metrics]
\label{def:metrics}
A \emph{port} is a pair \(p=(a,d)\), where \(a\in V(G)\) and
\(d\in\{\HH,\VV\}\) is an axis.  The \emph{literal last-edge distance}
\(b^{\mathrm{last}}_p(v)\) is the minimum number of internal bends in a
nonzero simple \(v\)-to-\(a\) path whose last actual grid edge has axis
\(d\).  If no such path exists, set the value to \(+\infty\); separately
set \(b^{\mathrm{last}}_p(a)=0\).

For the \emph{augmented distance}, attach at \(a\) a formal, zero-length
terminal stub of axis \(d\).  If \(P\) is a nonzero \(v\)-to-\(a\) grid
path, define its augmented bend count as its number of internal bends plus
one when the last edge of \(P\) is orthogonal to \(d\).  The extra term is
the bend between the last actual edge and the virtual stub.  Set the count
of the zero-length path at \(a\) to zero, and define
\[
  b^+_p(v)=
  \min\{\text{augmented bend count of a \(v\)-to-\(a\) path}\}.
\]
Unless explicitly stated otherwise, every distance, level, sector, and
sector graph in the rest of the paper uses the augmented metric, and we
abbreviate \(b(v)=b^+_p(v)\).
\end{definition}

\begin{proposition}[The endpoint convention changes a unit square]
\label{prop:endpoint}
Let \(X=[0,1]^2\), \(a=(0,0)\), and \(p=(a,\HH)\).  In the order
lower-left, lower-right, upper-right, upper-left, the literal last-edge
values are
\[
  (0,0,1,2),
\]
whereas the augmented values are
\[
  (0,0,1,1).
\]
In particular, the literal values are not the corner values of an affine
function on the square, and the literal metric changes by two across the
left edge.
\end{proposition}

\begin{proof}
For the literal metric, \((1,0)\) reaches \(a\) along one horizontal edge,
and \((1,1)\) uses the vertical edge to \((1,0)\) followed by a horizontal
edge, giving values \(0\) and \(1\).  A path from \((0,1)\) whose last
actual edge is horizontal must use
\[
  (0,1),(1,1),(1,0),(0,0);
\]
it has two internal bends, and no shorter path has the required last axis.
For the augmented metric, the single vertical edge from \((0,1)\) to \(a\)
is followed by the horizontal virtual stub and therefore costs one bend.
Finally,
\[
  0+1\neq 0+2,
\]
so the literal corner values violate the affine compatibility equation.
\end{proof}

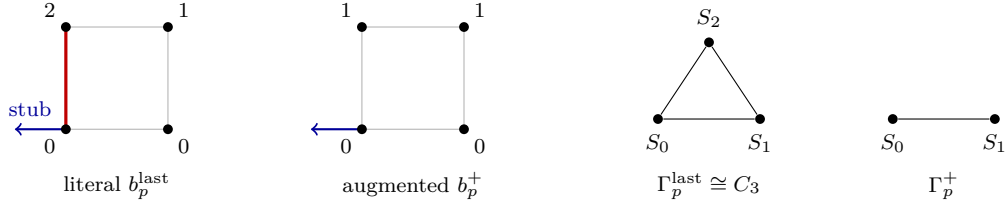
\begin{figure}[t]
\centering
\begin{tikzpicture}[scale=1.35,dot/.style={circle,fill,inner sep=1.3pt}]
\begin{scope}
  \draw[gray!60] (0,0) rectangle (1,1);
  \draw[very thick,red!75!black] (0,0) -- (0,1);
  \draw[->,thick,blue!60!black] (0,0) -- (-0.5,0);
  \node[blue!60!black,font=\scriptsize,above] at (-0.35,0.03) {stub};
  \node[dot] at (0,0) {}; \node[dot] at (1,0) {};
  \node[dot] at (1,1) {}; \node[dot] at (0,1) {};
  \node[font=\scriptsize,below left] at (0,0) {$0$};
  \node[font=\scriptsize,below right] at (1,0) {$0$};
  \node[font=\scriptsize,above right] at (1,1) {$1$};
  \node[font=\scriptsize,above left] at (0,1) {$2$};
  \node[font=\scriptsize] at (0.5,-0.55) {literal $b^{\mathrm{last}}_p$};
\end{scope}
\begin{scope}[xshift=2.9cm]
  \draw[gray!60] (0,0) rectangle (1,1);
  \draw[->,thick,blue!60!black] (0,0) -- (-0.5,0);
  \node[dot] at (0,0) {}; \node[dot] at (1,0) {};
  \node[dot] at (1,1) {}; \node[dot] at (0,1) {};
  \node[font=\scriptsize,below left] at (0,0) {$0$};
  \node[font=\scriptsize,below right] at (1,0) {$0$};
  \node[font=\scriptsize,above right] at (1,1) {$1$};
  \node[font=\scriptsize,above left] at (0,1) {$1$};
  \node[font=\scriptsize] at (0.5,-0.55) {augmented $b^{+}_p$};
\end{scope}
\begin{scope}[xshift=5.8cm,yshift=0.1cm]
  \node[dot,label={[font=\scriptsize]below:$S_0$}] (a) at (0,0) {};
  \node[dot,label={[font=\scriptsize]below:$S_1$}] (b) at (1,0) {};
  \node[dot,label={[font=\scriptsize]above:$S_2$}] (c) at (0.5,0.75) {};
  \draw (a)--(b)--(c)--(a);
  \node[font=\scriptsize] at (0.5,-0.65)
    {$\Gamma^{\mathrm{last}}_p\cong C_3$};
\end{scope}
\begin{scope}[xshift=8.1cm,yshift=0.1cm]
  \node[dot,label={[font=\scriptsize]below:$S_0$}] (a) at (0,0) {};
  \node[dot,label={[font=\scriptsize]below:$S_1$}] (b) at (1,0) {};
  \draw (a)--(b);
  \node[font=\scriptsize] at (0.5,-0.65) {$\Gamma^{+}_p$};
\end{scope}
\end{tikzpicture}
\caption{The endpoint counterexample
(\cref{prop:endpoint,cor:literal-c3}).  The port
\(p=((0,0),\HH)\) is drawn as a horizontal virtual stub.  Under the
literal last-edge reading the corner values are \((0,0,1,2)\); the left
edge (red) changes the value by two, and the literal sector graph is a
triangle although the square has no hole.  Under the augmented
convention the values are \((0,0,1,1)\) and the sector graph is a single
edge.}
\label{fig:endpoint}
\end{figure}

\begin{corollary}[The literal metric violates the tree theorem and the
hole bound]
\label{cor:literal-c3}
In the setting of \cref{prop:endpoint}, define literal sectors as
connected components of equal \(b^{\mathrm{last}}_p\)-value and let
\(\Gamma^{\mathrm{last}}_p\) be the graph with one vertex per literal
sector and an edge whenever a grid edge joins two distinct sectors.  The
literal sectors are
\[
  S_0=\{(0,0),(1,0)\},\qquad
  S_1=\{(1,1)\},\qquad
  S_2=\{(0,1)\},
\]
and
\[
  \Gamma^{\mathrm{last}}_p\cong C_3,
  \qquad
  \betti(\Gamma^{\mathrm{last}}_p)=1>0=h.
\]
In particular, neither the single-port tree theorem nor the cycle bound
\(\betti(\Gamma^{\mathrm{last}}_p)\leq h\) extends to the literal
last-edge metric: the augmented convention used in the remainder of the
paper is necessary, not merely convenient.
\end{corollary}

\begin{proof}
By \cref{prop:endpoint}, the literal values at
\((0,0),(1,0),(1,1),(0,1)\) are \(0,0,1,2\).  The bottom edge joins the
two value-zero vertices, so they form one sector \(S_0\), while
\((1,1)\) and \((0,1)\) are singleton sectors.  The right edge joins
\(S_0\) to \(S_1\) and the top edge joins \(S_1\) to \(S_2\).  The left
edge joins \(S_2\) directly to \(S_0\): its endpoint values are \(2\)
and \(0\).  These are three edges on three vertices, so
\(\Gamma^{\mathrm{last}}_p\) is a triangle.  The filled unit square is a
disk, so \(h=0\).
\end{proof}

\begin{proposition}[Operational layering computes the augmented metric]
\label{prop:layering}
Let \(L_0\) be the maximal straight grid path of axis \(d\) through \(a\).
For \(i\geq0\), let \(L_{i+1}\) consist of all as-yet unlabelled vertices
that are reachable by a straight grid path from a vertex of \(L_i\).  Then
\[
  L_i=\{v\in V(G):b^+_p(v)=i\}.
\]
\end{proposition}

\begin{proof}
Read a path from the virtual stub outward from \(a\).  A path with \(i\)
augmented bends is a concatenation of \(i+1\) straight pieces, the first of
which has axis \(d\).  Here the first piece may be only the formal
zero-length stub at \(a\); it records axis \(d\) before the first actual
edge leaves \(a\).  Its successive pieces show inductively that its
endpoint is labelled by stage at most \(i\).  Conversely, the straight
paths witnessing membership in \(L_0,L_1,\ldots,L_i\) concatenate to a grid
walk with at most \(i\) augmented bends.  Erasing closed subwalks, as in the
proof of \cref{lem:state}, does not increase the augmented bend count and
produces a simple path.  If a vertex newly assigned to \(L_i\) had augmented
distance less than \(i\), the first implication would have assigned it at an
earlier stage.  Thus its distance is exactly \(i\).
\end{proof}

\begin{remark}[Relation to the CMP terminology]
The formal arrival-axis wording in \cite{deligkas2026polygon} admits the
literal reading of \cref{prop:endpoint}, while the operational procedure
stated later in that work is precisely the procedure of
\cref{prop:layering}.  By \cref{cor:literal-c3}, the literal reading is
inconsistent with the single-port tree lemma of
\cite{deligkas2026polygon} already on a unit square, so the operational
reading is the only one under which that lemma can hold.  Our results
concern the latter, augmented metric; they do not assert an affine/Reeb
theorem for \(b^{\mathrm{last}}_p\).
\end{remark}

\begin{definition}[Sector graph]
For \(m\in\Z_{\ge 0}\), a \emph{level-\(m\) sector} is a connected component
of the induced subgraph \(G[\{v:b(v)=m\}]\).  The simple graph
\(\GammaP\) has one vertex for each sector and an edge between two distinct
sectors whenever a grid edge has one endpoint in each.
\end{definition}

Because \(G\) is connected, \(\GammaP\) is connected.  We prove formally in
\cref{lem:adjacent-levels} that different adjacent sectors have consecutive
labels.

For a connected finite graph or multigraph \(J\), we use
\[
  \betti(J)=|E(J)|-|V(J)|+1
\]
for its cycle rank, counting parallel edges with multiplicity.

\section{An exact orientation-state representation}
\label{sec:state}

This section introduces the technical device on which every later proof
runs: a two-orientation \(0/1\)-weighted state graph that computes the
augmented metric and exposes its parity structure.

Construct an undirected \(0/1\)-weighted graph \(\widehat G_p\) as follows.
For every \(v\in V(G)\), create states \((v,\HH)\) and \((v,\VV)\).
A horizontal grid edge \(uv\) creates a weight-zero edge
\((u,\HH)(v,\HH)\); a vertical grid edge creates a weight-zero edge
\((u,\VV)(v,\VV)\).  Finally, every \(v\) creates a weight-one switch edge
\((v,\HH)(v,\VV)\).  Let
\[
  \delta_{\HH}(v)=
  \dist_{\widehat G_p}\bigl((a,d),(v,\HH)\bigr),
  \qquad
  \delta_{\VV}(v)=
  \dist_{\widehat G_p}\bigl((a,d),(v,\VV)\bigr).
\]
All these distances are finite because \(G\) is connected and every pair
of orientation states over one vertex is joined by a switch edge.

\begin{lemma}[Augmented state representation]
\label{lem:state}
For every \(v\in V(G)\),
\[
  b^+_p(v)=\min\{\delta_{\HH}(v),\delta_{\VV}(v)\}.
\]
Furthermore:
\begin{enumerate}[label=(\alph*),leftmargin=2em]
  \item
  \(\lvert\delta_{\HH}(v)-\delta_{\VV}(v)\rvert=1\);
  \item
  if \(uv\) is horizontal, then
  \(\delta_{\HH}(u)=\delta_{\HH}(v)\);
  \item
  if \(uv\) is vertical, then
  \(\delta_{\VV}(u)=\delta_{\VV}(v)\).
\end{enumerate}
\end{lemma}

\begin{proof}
Reverse a grid path from \(v\) to the port together with its formal stub.
Every straight segment is represented by weight-zero edges within one
orientation layer, every internal bend by one switch edge, and a mismatch
between the last actual edge and the stub by a switch at \(a\).  Conversely,
project a state-graph walk to \(G\) and remove stationary switches away from
the changes of orientation; the switch at \(a\), if present, records exactly
the virtual-stub contribution.  The projection is a grid walk with the same
augmented bend count.

It remains to justify passage from a walk to a simple path.  Include the
stub axis \(d\) at the beginning of the reversed edge-axis word.  If the
walk visits one grid vertex twice, delete the closed subwalk between two
such visits.  At the splice, the new contribution satisfies
\[
  \mathbf 1\{\alpha\neq\beta\}
  \leq
  \sum_{r=i}^{j}\mathbf 1\{\alpha_r\neq\alpha_{r+1}\},
\]
where the sum is the sequence of axis changes formerly connecting the
incoming axis \(\alpha\) to the outgoing axis \(\beta\).  Thus deletion does
not increase augmented cost.  Repeating this operation produces a simple
path of no greater cost.  This proves the formula for \(b^+_p\).

The two states above \(v\) are joined by a weight-one edge, so their
distances differ by at most one.  Every walk from the source orientation
\(d\) to the same orientation uses an even number of switch edges, while
every walk to the opposite orientation uses an odd number.  Hence the two
distances have opposite parity and cannot be equal; their difference is
exactly one.  A horizontal grid edge gives weight-zero connections in both
directions in the \(\HH\)-layer, proving (b).  The vertical statement is
identical.
\end{proof}

\begin{lemma}[Adjacent levels]
\label{lem:adjacent-levels}
For every grid edge \(uv\),
\[
  |b(u)-b(v)|\leq 1.
\]
If \(u\) and \(v\) lie in different sectors, then their labels differ by
exactly one.
\end{lemma}

\begin{proof}
Suppose first that \(uv\) is horizontal.  By \cref{lem:state},
\[
  H:=\delta_{\HH}(u)=\delta_{\HH}(v).
\]
At either endpoint \(z\in\{u,v\}\), the vertical-state distance is
\(H-1\) or \(H+1\).  Therefore
\[
  b(z)=\min\{H,\delta_{\VV}(z)\}\in\{H-1,H\},
\]
which proves the inequality.  The proof for a vertical edge exchanges
\(\HH\) and \(\VV\).  If the two labels are equal, the edge \(uv\) itself
connects the endpoints inside one equal-label component, so they are in the
same sector.  Thus a cross-sector edge has consecutive integer labels.
\end{proof}

\begin{remark}
The state representation is also algorithmic.  A \(0\)-\(1\) breadth-first
search computes all values in time \(O(|V(G)|+|E(G)|)\).  No continuous
minimum-link metric is assumed.
\end{remark}

\section{The no-saddle square lemma and the canonical extension}
\label{sec:extension}

The parity structure of \cref{lem:state} forbids saddle patterns on unit
squares.  This forces a canonical piecewise-affine extension of the bend
labels---with no triangulation choices---whose level sets the next
section studies.

List the corner values of a unit square counterclockwise, starting at the
lower-left corner.

\begin{lemma}[No-saddle square]
\label{lem:nosaddle}
For every unit square \(Q\) of \(X\), there is an integer \(m\) such that the
four corner values of \(b\) have one of the following five forms:
\[
\begin{array}{c}
  (m,m,m,m),\\[1mm]
  (m,m,m+1,m+1),\qquad (m+1,m+1,m,m),\\[1mm]
  (m,m+1,m+1,m),\qquad (m+1,m,m,m+1).
\end{array}
\]
In particular, neither a checkerboard saddle nor a three-versus-one corner
pattern can occur.
\end{lemma}

\begin{proof}
It is enough to treat \(d=\HH\); exchanging the axes gives the other case.
By \cref{lem:state}, \(\delta_{\HH}\) is constant on the bottom edge of
\(Q\), with value \(h_0\), and on the top edge, with value \(h_1\).
Similarly, \(\delta_{\VV}\) has values \(v_0\) and \(v_1\) on the left and
right edges.  The \(h_i\) are even, the \(v_j\) are odd, and at each of the
four corners
\[
  |h_i-v_j|=1
  \qquad (i,j\in\{0,1\}).
\]

If \(h_0\neq h_1\), then parity and the displayed inequalities force
\(|h_0-h_1|=2\), while the only integer at distance one from both is their
midpoint.  Thus \(v_0=v_1=(h_0+h_1)/2\).  Taking
\(\min\{\delta_{\HH},\delta_{\VV}\}\) at each corner gives one of the two
horizontal stripe patterns.

If \(h_0=h_1=h_\star\), then each
\(v_j\in\{h_\star-1,h_\star+1\}\).  Equal choices give a constant pattern;
different choices give one of the two vertical stripe patterns.
\end{proof}

\begin{theorem}[Canonical cubical extension]
\label{thm:extension}
There is a unique continuous function \(f_p:X\to\R\) such that
\begin{enumerate}[label=(\alph*),leftmargin=2em]
  \item \(f_p(v)=b(v)\) at every grid vertex; and
  \item \(f_p\) is affine on every unit square of \(X\).
\end{enumerate}
On each square, \(f_p\) is either constant or changes linearly by one along
exactly one coordinate axis.
\end{theorem}

\begin{proof}
For the four corner values \(b_{00},b_{10},b_{11},b_{01}\) of any square,
\cref{lem:nosaddle} gives
\[
  b_{00}+b_{11}=b_{10}+b_{01}.
\]
This is exactly the compatibility condition for an affine function of two
variables to interpolate all four values.  The five patterns in
\cref{lem:nosaddle} give the stated local form.  Two adjacent squares have
the same endpoint values on their common edge, and their affine
restrictions to that edge therefore agree.  The local functions glue to a
continuous \(f_p\).  Uniqueness holds square by square.
\end{proof}

\section{Sectors as the simple shadow of a Reeb multigraph}
\label{sec:reeb}

We now extract the topology of the canonical extension: its integer
level sets are exactly the sectors, its open bands are products, and its
Reeb quotient is a finite multigraph of which the simple sector graph is
the parallel-edge shadow.

For a continuous \(f:X\to\R\), the Reeb space \(R_f\) is the quotient of
\(X\) that identifies two points precisely when they lie in the same
connected component of one level set \(f^{-1}(t)\).  We prove below, rather
than assume, that the quotient in our setting is a finite topological graph.

For an integer \(m\), set
\[
  L_m=f_p^{-1}(m),
  \qquad
  U_m=f_p^{-1}((m,m+1)).
\]

\begin{lemma}[Integer contours]
\label{lem:integer}
The connected components of \(L_m\) are in bijection with the level-\(m\)
sectors.
\end{lemma}

\begin{proof}
On a constant-\(m\) square, all of the square lies in \(L_m\), and its
boundary is already a connected subgraph of level-\(m\) grid edges.  On a
stripe square with endpoint values \(m\) and \(m+1\), the intersection with
\(L_m\) is exactly its low-valued side; for an \((m-1,m)\) stripe the same
holds on its high-valued side.  Hence \(L_m\) is obtained from the
level-\(m\) grid subgraph by filling constant squares whose boundaries are
already connected.  Such fillings neither merge distinct grid components
nor split a component.
\end{proof}

\begin{lemma}[Open bands are products]
\label{lem:bands}
Let \(C\) be a connected component of \(U_m\).  Then:
\begin{enumerate}[label=(\alph*),leftmargin=2em]
  \item for every \(t\in(m,m+1)\), the intersection
  \(C\cap f_p^{-1}(t)\) is one connected component of \(f_p^{-1}(t)\);
  \item the closure \(\overline C\) meets exactly one component of \(L_m\)
  and exactly one component of \(L_{m+1}\).
\end{enumerate}
\end{lemma}

\begin{proof}
Only \(m/(m+1)\)-stripe squares meet \(U_m\).  Within each such square,
\((Q\cap U_m,f_p)\) is a product
\(I_Q\times(m,m+1)\), with the second coordinate equal to the function
value.  Two stripe squares glue across a varying edge, and that gluing is
the same at every \(t\in(m,m+1)\).  Gluing the intervals \(I_Q\) along these
varying edges produces a finite one-dimensional complex \(K\), and each
connected component of \(U_m\) has the form
\[
  K_0\times(m,m+1)
\]
for a connected component \(K_0\) of \(K\).  This proves (a).

When adjacent stripe squares glue along a varying edge, their low endpoints
are the same level-\(m\) grid vertex and their high endpoints are the same
level-\((m+1)\) grid vertex.  Propagating along a connected chain of stripe
squares shows that all low sides of \(C\) lie in one component of \(L_m\),
and all high sides lie in one component of \(L_{m+1}\).  Each side is
nonempty, proving (b).
\end{proof}

Define a finite multigraph \(\calR\) as follows.  Its vertices are all
components of all nonempty \(L_m\).  Every component \(C\) of an open band
\(U_m\) contributes one edge between the two integer contours met by
\(\overline C\), as specified by \cref{lem:bands}.  Different band
components may contribute parallel edges.

\begin{theorem}[Reeb realization and simple shadow]
\label{thm:shadow}
The geometric realization of \(\calR\) is homeomorphic to \(\Reeb\).
Under the bijection of \cref{lem:integer}, the simple sector graph
\(\GammaP\) is obtained from \(\calR\) by merging every nonempty family of
parallel edges into one edge.  Therefore
\[
  \betti(\GammaP)\leq \betti(\Reeb).
\]
\end{theorem}

\begin{proof}
We first establish the topology rather than assuming in advance that the
Reeb quotient is a graph.  Let \(|\calR|\) be the geometric realization of
the finite multigraph just constructed.  Define
\[
  q:X\longrightarrow|\calR|
\]
cell by cell.  A constant integer-valued square is sent to the vertex
corresponding to its integer contour.  If a stripe square belongs to a
component \(C\subseteq U_m\), send its entire closed square to the closed
edge \(e_C\), with coordinate \(f_p(x)-m\in[0,1]\).  Two stripe squares
sharing a varying edge belong to the same band component, so their maps
agree there.  On a shared constant edge, both maps take the value of the
corresponding integer-contour vertex; the same is true at a shared integer
vertex.  The maps therefore agree on all intersections in the finite
closed-cell cover of \(X\).  The finite gluing lemma gives a continuous map
\(q\).

The fibers of \(q\) are exactly the connected components of level sets of
\(f_p\): this is \cref{lem:integer} at integer values and
\cref{lem:bands} at non-integer values.  Thus \(q\) factors through a
continuous bijection
\[
  \bar q:\Reeb\longrightarrow|\calR|.
\]
The Reeb quotient is compact because it is a quotient of the compact
space \(X\), while the finite graph \(|\calR|\) is Hausdorff.  Therefore
\(\bar q\) is a homeomorphism.  This proves at the same time, without
circularity, that \(\Reeb\) is a finite topological multigraph.

It remains to compare adjacencies.  Suppose a grid edge joins a level-\(m\)
sector to a level-\((m+1)\) sector.  It is a varying side of at least one
stripe square, so a component of \(U_m\) has those contours as its
endpoints.  Conversely, every component of \(U_m\) contains a stripe
square, whose varying side is a grid edge joining its two endpoint sectors.
Thus \(\calR\) and \(\GammaP\) have the same adjacent vertex pairs; only
multiplicity is forgotten in \(\GammaP\).  Merging parallel edges cannot
increase cycle rank.
\end{proof}

\begin{remark}
The word ``shadow'' in \cref{thm:shadow} is substantive.  Even when
\(\GammaP\) is a tree, \(\Reeb\) may have a cycle carried by two distinct
open-band components with the same endpoint sectors; see
\cref{ex:parallel}.
\end{remark}

\section{A hole bound and a computable feedback set}
\label{sec:bound}

With the Reeb realization in hand, the cycle bound follows from general
topology in three steps: Gelbukh's cycle-rank theorem, a corank
estimate, and Alexander duality.

Gelbukh proved that if \(Y\) is connected and locally path-connected and the
Reeb graph \(R_f\) of a continuous function is a finite graph, then
\[
  \betti(R_f)\leq\corank\pi_1(Y),
\]
where the corank is the maximum rank of a free-group quotient
\cite[Theorem 3.1]{gelbukh2019approximation}.  The hypotheses hold for the
finite cubical complex \(X\) and for \(f_p\) by \cref{thm:shadow}.

\begin{theorem}[Cycle rank is bounded by holes]
\label{thm:holebound}
For the augmented metric of every port \(p\) in a finite connected pure
planar cubical domain with \(h\) holes,
\[
  \boxed{\betti(\GammaP)\leq h.}
\]
More precisely,
\[
  \betti(\GammaP)
  \leq \betti(\Reeb)
  \leq \corank\pi_1(X)
  \leq \operatorname{rank}H_1(X;\Z)
  =h.
\]
\end{theorem}

\begin{proof}
The first inequality is \cref{thm:shadow}; the second is Gelbukh's theorem.
If \(\pi_1(X)\twoheadrightarrow F_r\), then abelianization gives a
surjection \(H_1(X;\Z)\twoheadrightarrow\Z^r\).  Hence
\(\corank\pi_1(X)\leq\operatorname{rank}H_1(X;\Z)\).

Finally, \(X\) is compact and locally contractible.  Alexander duality
\cite[Corollary 3.45]{hatcher2002algebraic} gives
\[
  \widetilde H^1(X;\Z)
  \cong \widetilde H_0(S^2\setminus X;\Z).
\]
The complement has \(h+1\) components, so the group on the right has rank
\(h\).  The universal coefficient theorem, together with the fact that
\(H_0(X;\Z)\) is free, gives
\[
  \operatorname{rank}H^1(X;\Z)
  =\operatorname{rank}H_1(X;\Z).
\]
Hence \(\operatorname{rank}H_1(X;\Z)=h\), as required.
\end{proof}

When \(h=0\), \cref{thm:holebound} recovers acyclicity of the single-port
sector graph.  By itself it does not recover the straight-baseline claims
used in the hole-free CMP proof; those are geometric rather than purely
graph-theoretic.  \Cref{sec:examples} shows that with one hole no such
geometric recovery is possible in general; \cref{sec:hf-single-port}
then shows that on cubical disks the multiplicity-sensitive Reeb tree
and state parity recover an exact one-sided extrusion theorem.

\begin{lemma}[Feedback vertices]
\label{lem:fvs}
Every connected finite simple graph \(J\) has a feedback vertex set \(D\)
with
\[
  |D|\leq \betti(J).
\]
Given \(J\), such a set can be found in linear time.
\end{lemma}

\begin{proof}
Choose a spanning tree \(T\).  For every cotree edge
\(e\in E(J)\setminus E(T)\), put one arbitrary endpoint of \(e\) into
\(D\), discarding duplicates.  Then
\[
  |D|\leq |E(J)\setminus E(T)|=\betti(J).
\]
Every edge outside \(T\) is incident to \(D\), so \(J-D\) is a subgraph of
\(T\) and hence a forest.
\end{proof}

\begin{corollary}[Computable topological defects]
\label{cor:algorithm}
For a fixed port, one can compute in
\(O(|V(G)|+|E(G)|)\) time the sector graph and a set
\(D_p\subseteq V(\GammaP)\) such that
\[
  |D_p|\leq h,
  \qquad
  \GammaP-D_p\text{ is a forest}.
\]
\end{corollary}

\begin{proof}
Use \(0\)-\(1\) breadth-first search in \(\widehat G_p\), compute connected
components within each integer label, build \(\GammaP\), and apply
\cref{lem:fvs}.  The state graph and all subsequent graphs have linear size.
\end{proof}

\section{Three one-hole obstructions}
\label{sec:examples}

\Cref{thm:holebound} is what survives with holes.  This section shows
that it is essentially \emph{all} that survives: three coordinate-level
examples, each with a single hole, defeat the three geometric
strengthenings one might hope to import from the hole-free theory.  The
first shows that the simple sector graph can hide Reeb cycles, so the
topology cannot be read off the quotient alone; the second defeats
unique predecessors and port-side separators; the third defeats straight
baselines even when the quotient is a tree.

Except for the literal endpoint counterexample in
\cref{prop:endpoint}, all bend labels in this section use the augmented
metric from \cref{def:metrics}.

\subsection{The sector graph need not equal the Reeb graph}

\begin{example}[A Reeb cycle erased by simplification]
\label{ex:parallel}
Let
\[
  X=[0,3]^2\setminus
  \operatorname{int}\bigl([1,2]\times[1,2]\bigr),
  \qquad
  p=((0,1),\HH).
\]
Equivalently, \(X\) is the union of the eight cells in a \(3\times3\)
cell square other than \(Q_{1,1}\).  It is a pure cubical domain with one
hole; see \cref{fig:parallel}.  At grid vertices the bend values are
\[
\begin{array}{c|cccc}
  &x=0&x=1&x=2&x=3\\ \hline
y=3&1&1&1&1\\
y=2&1&1&1&1\\
y=1&0&0&0&0\\
y=0&1&1&1&1
\end{array}
\]
There are three sectors:
\[
\begin{aligned}
  S_0&=\{(x,1):0\leq x\leq3\},\\
  S_+&=\{(x,y):0\leq x\leq3,\ y\in\{2,3\}\},\\
  S_-&=\{(x,0):0\leq x\leq3\}.
\end{aligned}
\]
Thus \(\GammaP\) is the path \(S_+-S_0-S_-\).

For every \(t\in(0,1)\), however, the portion of \(f_p^{-1}(t)\) above
\(S_0\) has two components, one on each side of the missing cell.  They
give two parallel edges between \(S_0\) and \(S_+\) in \(\calR\).
The lower band gives one edge between \(S_0\) and \(S_-\).  Hence
\[
  \betti(\GammaP)=0,
  \qquad
  \betti(\Reeb)=1=h.
\]
\end{example}

\begin{proof}[Verification]
Every vertex on \(y=1\) reaches the port horizontally without a bend.
Every vertex above or below that row reaches it by one vertical segment
followed by a horizontal segment.  This proves the table and the sector
description.  The open upper band crosses the cells \(Q_{0,1}\) and
\(Q_{2,1}\), which cannot connect through the absent \(Q_{1,1}\).
Both components have the same integer endpoint sectors.  The lower band is
connected.
\end{proof}

\begin{figure}[t]
\centering
\begin{tikzpicture}[scale=0.85,dot/.style={circle,fill,inner sep=1.1pt}]
\begin{scope}
  \fill[blue!10]  (0,2) rectangle (3,3);
  \fill[orange!25] (0,1) rectangle (1,2);
  \fill[orange!25] (2,1) rectangle (3,2);
  \fill[green!18] (0,0) rectangle (3,1);
  \fill[gray!50]  (1,1) rectangle (2,2);
  \draw[gray!55] (0,0) grid (3,3);
  \draw[thick] (1,1) rectangle (2,2);
  \draw[very thick,blue!60!black] (0,1) -- (3,1);
  \draw[->,very thick,blue!60!black] (0,1) -- (-0.55,1);
  \node[blue!60!black,font=\scriptsize,above] at (-0.4,1.02) {$p$};
  \foreach \x in {0,1,2,3} {
    \node[font=\tiny,below=1pt] at (\x,0) {$1$};
    \node[font=\tiny,above=1pt] at (\x,3) {$1$};
  }
  \node[font=\tiny,below left=0pt] at (0,1) {$0$};
  \node[font=\tiny,below right=0pt] at (3,1) {$0$};
  \node[font=\tiny,above left=0pt] at (0,2) {$1$};
  \node[font=\tiny,above right=0pt] at (3,2) {$1$};
\end{scope}
\begin{scope}[xshift=5.4cm,yshift=0.35cm]
  \node[dot,label={[font=\scriptsize]above:$S_+$}] (sp) at (0.5,2.2) {};
  \node[dot,label={[font=\scriptsize]left:$S_0$}]  (s0) at (0.5,1.1) {};
  \node[dot,label={[font=\scriptsize]below:$S_-$}] (sm) at (0.5,0) {};
  \draw (s0) to[bend left=40] (sp);
  \draw (s0) to[bend right=40] (sp);
  \draw (s0) -- (sm);
  \node[font=\scriptsize] at (0.5,-0.7) {$\mathcal R_p$};
\end{scope}
\begin{scope}[xshift=7.6cm,yshift=0.35cm]
  \node[dot,label={[font=\scriptsize]above:$S_+$}] (sp) at (0.5,2.2) {};
  \node[dot,label={[font=\scriptsize]left:$S_0$}]  (s0) at (0.5,1.1) {};
  \node[dot,label={[font=\scriptsize]below:$S_-$}] (sm) at (0.5,0) {};
  \draw (s0) -- (sp);
  \draw (s0) -- (sm);
  \node[font=\scriptsize] at (0.5,-0.7) {$\GammaP$};
\end{scope}
\end{tikzpicture}
\caption{\Cref{ex:parallel}: a \(3\times3\) cell square with the middle
cell removed and \(p=((0,1),\HH)\).  Vertex labels show the augmented
bend distance; the level-zero contour is the thick row.  The open band
between levels \(0\) and \(1\) above that row has two components
(orange), one on each side of the hole, giving parallel edges of the
Reeb multigraph \(\mathcal R_p\); the lower band (green) is connected.
The simple sector graph \(\GammaP\) merges the parallel edges and hides
the cycle: \(\betti(\GammaP)=0<1=\betti(\Reeb)=h\).}
\label{fig:parallel}
\end{figure}
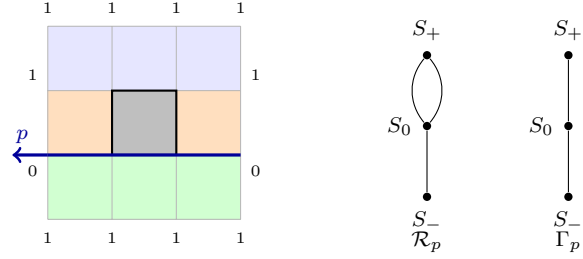

\subsection{A \texorpdfstring{\(C_4\)}{C4} and non-unique predecessors}

The next example uses the edge-discretization convention of
\cite{deligkas2026polygon}: retain grid vertices and unit grid edges fully
contained in a polygonal free domain.  It is deliberately not presented as
the one-skeleton of a pure unit-square complex, so the construction of
\cref{thm:extension} is not being invoked.  See \cref{fig:c4}.

\begin{example}[A one-hole \(C_4\)]
\label{ex:c4}
Let
\[
  P=
  \left[-\tfrac14,\tfrac{21}{4}\right]^2
  \setminus
  \left(\tfrac34,\tfrac54\right)^2.
\]
Its grid edge-discretization is the induced \(6\times6\) vertex grid
\[
  V(G_P)=\{0,\ldots,5\}^2\setminus\{(1,1)\}.
\]
Take the port \(p=((0,1),\VV)\).  The four sectors are
\[
\begin{aligned}
S_0&=\{(0,y):0\leq y\leq5\},
  &b(S_0)&=0,\\
S_1^-&=\{(x,0):1\leq x\leq5\},
  &b(S_1^-)&=1,\\
S_1^+&=\{(x,y):1\leq x\leq5,\ 2\leq y\leq5\},
  &b(S_1^+)&=1,\\
S_2&=\{(x,1):2\leq x\leq5\},
  &b(S_2)&=2.
\end{aligned}
\]
The sector labels, from \(y=5\) down to \(y=0\), are
\[
\begin{array}{c|cccccc}
 &0&1&2&3&4&5\\ \hline
5&S_0&S_1^+&S_1^+&S_1^+&S_1^+&S_1^+\\
4&S_0&S_1^+&S_1^+&S_1^+&S_1^+&S_1^+\\
3&S_0&S_1^+&S_1^+&S_1^+&S_1^+&S_1^+\\
2&S_0&S_1^+&S_1^+&S_1^+&S_1^+&S_1^+\\
1&S_0&\#&S_2&S_2&S_2&S_2\\
0&S_0&S_1^-&S_1^-&S_1^-&S_1^-&S_1^-
\end{array}
\]
The sector graph has exactly the edges
\[
  S_0S_1^-,
  \quad S_1^-S_2,
  \quad S_2S_1^+,
  \quad S_1^+S_0,
\]
and is therefore \(C_4\).
\end{example}

\begin{proof}[Verification]
The left column reaches \(p\) vertically with no bend.  A vertex in
\(S_1^-\) reaches \((0,0)\) horizontally and then reaches \(p\) vertically;
a vertex in \(S_1^+\) reaches the left column horizontally and then reaches
\(p\) vertically.  These routes have one bend, and zero bends are impossible
off the left column.

For \(v=(x,1)\) with \(x\geq2\), a route can go below or above the missing
vertex and then enter the left column; either route has two bends.  A
one-bend route whose final axis is vertical would have to reach the left
column by the horizontal segment along \(y=1\), which is blocked at
\((1,1)\).  Thus \(b(v)=2\).  The displayed grid then gives exactly the four
adjacencies.
\end{proof}

\begin{figure}[t]
\centering
\begin{tikzpicture}[scale=0.6,dot/.style={circle,fill,inner sep=0.9pt}]
\begin{scope}
  \fill[blue!14,rounded corners=2pt]  (-0.3,-0.3) rectangle (0.3,5.3);
  \fill[green!18,rounded corners=2pt] (0.7,-0.3) rectangle (5.3,0.3);
  \fill[orange!20,rounded corners=2pt](0.7,1.7)  rectangle (5.3,5.3);
  \fill[red!16,rounded corners=2pt]   (1.7,0.7)  rectangle (5.3,1.3);
  \draw[gray!45] (0,0) grid (5,5);
  \fill[white] (1,1) circle (0.22);
  \node[font=\scriptsize] at (1,1) {$\times$};
  \foreach \x in {0,...,5} \foreach \y in {0,...,5} {
    \ifnum\x=1 \ifnum\y=1 \else \node[dot] at (\x,\y) {}; \fi
    \else \node[dot] at (\x,\y) {}; \fi
  }
  \draw[->,very thick,blue!60!black] (-0.55,0.6) -- (-0.55,1.4);
  \node[blue!60!black,font=\scriptsize,left] at (-0.55,1) {$p$};
  \node[font=\scriptsize,blue!50!black]  at (0,5.75)   {$S_0$};
  \node[font=\scriptsize,green!40!black] at (5.9,0)    {$S_1^-$};
  \node[font=\scriptsize,orange!70!black]at (5.9,3.5)  {$S_1^+$};
  \node[font=\scriptsize,red!60!black]   at (5.9,1)    {$S_2$};
\end{scope}
\begin{scope}[xshift=9cm,yshift=1.2cm,scale=0.9]
  \node[circle,draw,fill=blue!14,inner sep=1.5pt,font=\scriptsize]
    (n0) at (0,1.5) {$S_0$};
  \node[circle,draw,fill=green!18,inner sep=1.5pt,font=\scriptsize]
    (n1) at (1.5,0) {$S_1^-$};
  \node[circle,draw,fill=red!16,inner sep=1.5pt,font=\scriptsize]
    (n2) at (3,1.5) {$S_2$};
  \node[circle,draw,fill=orange!20,inner sep=1.5pt,font=\scriptsize]
    (n3) at (1.5,3) {$S_1^+$};
  \draw (n0)--(n1)--(n2)--(n3)--(n0);
\end{scope}
\end{tikzpicture}
\caption{\Cref{ex:c4}: the edge-discretized \(6\times6\) grid with the
vertex \((1,1)\) removed and \(p=((0,1),\VV)\).  The four sectors are
shaded; the level-two sector \(S_2\) has the two distinct level-one
neighbours \(S_1^-\) and \(S_1^+\), and the sector graph is the cycle
\(C_4\).  Unique predecessors fail already with one hole.}
\label{fig:c4}
\end{figure}
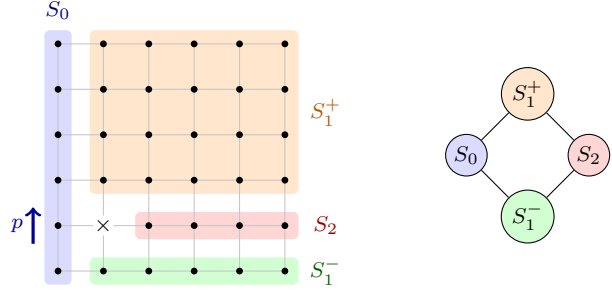

The level-two sector \(S_2\) has two distinct level-one neighbours,
\(S_1^-\) and \(S_1^+\).  Hence the unique-predecessor statement from the
hole-free single-port theory is false already with one hole.  Choosing
either incoming side cannot yield a separator from the port: the other side
provides a bypass.

\subsection{A tree quotient with no straight baseline}

We now return to the pure cubical setting.  This example shows that even
acyclicity is not enough to recover the geometric sector structure.

\begin{definition}[One-sided straight generation]
\label{def:baseline}
Call a vertex set \(S\) \emph{one-sided baseline generated} if there are a
straight grid path \(P\subseteq S\) and a cardinal direction
\(e\in\{\pm(1,0),\pm(0,1)\}\), perpendicular to \(P\), such that every
\(v\in S\) lies on a grid segment contained in \(S\), parallel to \(e\),
whose first point is in \(P\).  Any sector with the stronger directional
baseline property used in the hole-free CMP decomposition is one-sided
baseline generated.
\end{definition}

If \(e\) is horizontal, every horizontal grid-line section of a one-sided
baseline-generated set is an interval.  If \(e\) is vertical, every
vertical grid-line section is an interval.

\begin{example}[Tree sector graph, no baseline]
\label{ex:nobaseline}
Let \(B=[0,10]^2\), tiled into unit cells.  Remove the interior of the
polyomino formed by the nine cells
\[
\begin{split}
\mathcal O=\{&
(3,4),(4,3),(4,4),(4,5),(4,6),\\
&(5,3),(5,4),(5,6),(6,6)\},
\end{split}
\]
where \((x,y)\) denotes \(Q_{x,y}\).  Let \(X\) be the cubical complex
formed by all remaining cells and their faces.  The obstacle boundary,
after suppressing collinear points, is
\[
\begin{split}
(3,4),(4,4),(4,3),(6,3),(6,5),(5,5),\\
(5,6),(7,6),(7,7),(4,7),(4,5),(3,5),
\end{split}
\]
and then back to \((3,4)\).  This is a simple rectilinear cycle, so \(X\)
has one hole.

Take \(p=((0,10),\HH)\).  There are exactly three sectors.  The level-zero
sector \(S_0\) is the top row.  The unique level-two sector is
\[
\begin{split}
C={}&\{4,5,6\}\times\{0,1,2,3\}\\
&\cup\{(4,4),(6,4),(5,5),(6,5),(5,6),(6,6)\}.
\end{split}
\]
All other vertices have label one and form one sector \(S_1\).  Thus
\[
  \GammaP=S_0-S_1-C
\]
is a tree.
\end{example}

\begin{proof}[Sector verification]
The port axis is horizontal, so the top row has bend distance zero.  A
vertex below the top row has bend distance one exactly when its upward
vertical segment to the top row is unobstructed: that segment followed by
the top-row segment to \(p\) has one bend, and every one-bend route must
have this form.  Inspecting the nine obstacle cells shows that the vertices
whose upward segment is blocked are precisely those in \(C\).  Each such
vertex can first move horizontally to one side of the obstacle, then
vertically to the top row, and finally horizontally to \(p\), using two
bends.  Hence the labels are as claimed.

The sets \(S_0,S_1,C\) are each connected.  Edges join \(S_0\) to \(S_1\)
and \(S_1\) to \(C\), while no edge can join levels zero and two.  Therefore
the sector graph is the displayed three-vertex path.
\end{proof}

\begin{proposition}[The last sector has no baseline]
\label{prop:nobaseline}
The sector \(C\) in \cref{ex:nobaseline} is not one-sided baseline
generated.  In particular, it has no directional baseline of the type used
in the hole-free sector decomposition.
\end{proposition}

\begin{proof}
The horizontal section at \(y=4\) is
\[
  C\cap(\Z\times\{4\})=\{(4,4),(6,4)\},
\]
which is not an interval because \((5,4)\notin V(G)\).  Thus \(C\) cannot
be generated in either horizontal direction.  The vertical section at
\(x=5\) is
\[
  C\cap(\{5\}\times\Z)
  =\{(5,y):y\in\{0,1,2,3,5,6\}\},
\]
which is also disconnected.  Thus \(C\) cannot be generated in either
vertical direction.
\end{proof}

\cref{ex:nobaseline} has \(\betti(\GammaP)=0\), so the feedback set in
\cref{cor:algorithm} may be empty even though a geometric defect is present.
Topological defects and baseline defects are therefore different objects.

\section{Exact one-sided geometry in the hole-free case}
\label{sec:hf-single-port}

The previous section showed which statements a single hole destroys.  We
now reverse the lens and show that in the hole-free full-cell setting
those statements are not merely true but exact.
The cycle bound becomes substantially more rigid when the cubical domain is
a disk.  In this section we assume, in addition to the standing purity
hypothesis, that \(\lvert X\rvert\) is homeomorphic to a closed disk.  Thus
there are no point pinches.  We continue to use the augmented, virtual-stub
metric of \cref{def:metrics}.  The result below is deliberately restricted
to this full-cell setting; it does not assert the same geometry for a
discretized-polygon graph with dangling one-dimensional edges.

For an axis \(\tau\in\{\HH,\VV\}\), write \(\bar\tau\) for the other
axis.  For \(m\geq0\), define the \emph{preferred axis}
\[
  \eta_m=
  \begin{cases}
    d,&m\text{ even},\\
    \bar d,&m\text{ odd}.
  \end{cases}
\]

\begin{lemma}[Preferred axes and cross-level edges]
\label{lem:hf-preferred-axis}
If \(b(v)=m\), then
\[
  \delta_{\eta_m}(v)=m,
  \qquad
  \delta_{\bar\eta_m}(v)=m+1.
\]
If a grid edge has endpoint labels \(m-1\) and \(m\), its axis is
\(\eta_m\).
\end{lemma}

\begin{proof}
By \cref{lem:state}, the two state distances at \(v\) differ by one.
Moreover, the distance to the source axis has even parity and the distance
to the other axis has odd parity.  The smaller state distance is \(m\), so
its parity uniquely determines its axis and gives the first assertion.

Let \(uv\) be a cross-level edge, with \(b(u)=m-1\) and \(b(v)=m\).
If its axis were the preferred axis at level \(m-1\), the corresponding
state distance would be constant across \(uv\), by
\cref{lem:state}.  The displayed formula would instead give the two values
\(m-1\) at \(u\) and \(m+1\) at \(v\), a contradiction.  Hence the edge
axis is \(\eta_m\).
\end{proof}

\begin{corollary}[The Reeb multigraph is a tree]
\label{cor:hf-reeb-tree}
The Reeb multigraph \(\calR\) of \cref{thm:shadow} is a tree.  It has no
parallel edges, and hence it is identical, as an abstract graph, to the
single-port sector graph \(\GammaP\).
\end{corollary}

\begin{proof}
Because \(X\) is a disk, \(\pi_1(X)\) is trivial.  The proof of
\cref{thm:holebound}, before parallel edges are merged, gives
\[
  \betti(\calR)=\betti(\Reeb)
  \leq\corank\pi_1(X)=0.
\]
The Reeb quotient is connected, so the finite multigraph \(\calR\) is a
tree.  In particular it has no parallel edges.  The simple-shadow
identification in \cref{thm:shadow} therefore performs no merger.
\end{proof}

\begin{lemma}[Every open band is a straight strip]
\label{lem:hf-straight-band}
Let \(C\) be a connected component of
\(f_p^{-1}((m-1,m))\), where \(m>0\).  There are two equal-length,
parallel straight grid paths
\[
  P^-\subseteq f_p^{-1}(m-1),
  \qquad
  P^+\subseteq f_p^{-1}(m),
\]
and a cardinal unit direction \(e\), of axis \(\eta_m\), such that
\[
  P^+=P^-+e.
\]
The stripe squares meeting \(C\) form exactly the unit-width straight
strip between \(P^-\) and \(P^+\).  Its cross edges are the
order-preserving perfect matching between the two shores.
\end{lemma}

\begin{proof}
Every square meeting \(C\) is an \((m-1)/m\)-stripe square.  Each of its
varying grid edges is a cross-level edge and therefore has axis \(\eta_m\)
by \cref{lem:hf-preferred-axis}.  Two open stripe pieces can meet only
across a varying edge: a constant side, and every corner, has integer
function value and is absent from the open band.

Consequently, two adjacent stripe squares in \(C\) lie in the same unit
coordinate slab and sit side by side in the direction orthogonal to
\(\eta_m\).  Their low-to-high directions agree on their shared varying
edge.  Propagating along a connected chain shows that all squares meeting
\(C\) lie in one row or one column, with one common low-to-high direction.
A finite connected set of unit squares in that row or column is a
contiguous interval.  Its two constant sides are the required paths
\(P^-\) and \(P^+\), and the product structure gives the stated matching.

Purity is used here: every changing grid edge belongs to a square, so
there is no additional one-dimensional open-band case.
\end{proof}

\begin{theorem}[Augmented single-port geometry on a cubical disk]
\label{thm:hf-osp}
The sector graph \(\GammaP\), rooted at its level-zero sector, has the
following properties.
\begin{enumerate}[label=(\alph*),leftmargin=2em]
  \item
  The root is the maximal straight grid path through \(a\) of axis \(d\).
  Every level-\(m\) sector has rooted depth \(m\).
  \item
  Every level-\(m\) sector \(S\), \(m>0\), has a unique predecessor
  \(A\) of level \(m-1\).
  \item
  There are equal-length straight paths
  \(\overline P_S\subseteq A\) and \(P_S\subseteq S\), and a cardinal
  direction \(e_S\) of axis \(\eta_m\), such that all \(A\)--\(S\) grid
  edges are precisely the order-preserving matching
  \[
    x(x+e_S),
    \qquad
    x\in\overline P_S,
    \qquad
    P_S=\overline P_S+e_S.
  \]
  \item
  The sector is exactly the one-sided directional extrusion
  \[
    \boxed{
    S=\operatorname{Ext}_{e_S}(P_S)
    :=
    \{x+re_S:
      x\in P_S,\ r\in\Z_{\geq0},\
      x,x+e_S,\ldots,x+re_S\text{ is a path in }G\}.
    }
  \]
  \item
  Let \(\mathcal T_S\) be the rooted subtree of \(\GammaP\) below \(S\),
  and let \(W_S\) be the union of the vertices in its sectors.  Every grid
  path from \(W_S\setminus P_S\) to \(V(G)\setminus W_S\) meets \(P_S\).
  Thus \(P_S\) is a root-side gate; if
  \(W_S\setminus P_S\neq\varnothing\), it is also a nontrivial vertex
  separator.
\end{enumerate}
\end{theorem}

\begin{proof}
Zero state cost permits movement only along the stub axis, so the
level-zero vertices form exactly the maximal straight \(d\)-path through
\(a\).  This gives the unique root.

We first prove that every positive sector has a lower-level neighbor.
Take \(v\) in a level-\(m\) sector \(S\), and among the shortest
state-graph paths from \((a,d)\) to the preferred state
\((v,\eta_m)\), choose one using the fewest zero-weight movement edges.
Its last switch
raises the path cost from \(m-1\) to \(m\), after which it follows a
straight \(\eta_m\)-run.  Prefixes immediately before and immediately
after the switch are shortest paths to their respective states; hence the
switch vertex has scalar label \(m-1\).  The final run has positive
length, since otherwise its endpoint would have scalar label at most
\(m-1\).  Along this zero-weight run, the \(\eta_m\)-state distance is
constantly \(m\);
therefore the scalar labels on the run are only \(m-1\) and \(m\).
The last transition from label \(m-1\) into the component containing
\(v\) supplies a lower-level neighbor of \(S\).

By \cref{cor:hf-reeb-tree}, \(\GammaP\) is a tree.  Repeatedly choosing a
lower neighbor reaches the root in exactly \(m\) steps.  The root path in
a tree is unique, so the lower neighbor is unique and the rooted depth is
\(m\).  This proves (a) and (b).

Let \(A\) be the predecessor of \(S\).  Every \(A\)--\(S\) grid edge lies
on an \((m-1)/m\)-stripe square and hence in an open band joining the two
integer contours.  There can be only one such band, since two would be
parallel \(A\)--\(S\) edges in the Reeb multigraph, contradicting
\cref{cor:hf-reeb-tree}.  Applying \cref{lem:hf-straight-band} to this
unique band gives \(\overline P_S,P_S\), and \(e_S\), and proves (c).

For the first inclusion in (d), again take \(v\in S\), and among all
shortest state paths to \((v,\eta_m)\), choose one with the fewest
zero-weight movement edges.  Its final run is a simple monotone segment
of one grid line.  On that run, take the last edge \(xy\) directed from
label \(m-1\) to label \(m\) before \(v\).  The
suffix from \(y\) to \(v\) has constant label \(m\), so it lies in \(S\).
The lower endpoint \(x\) lies in the unique predecessor \(A\), and hence
\(xy\) belongs to the unique interface of (c).  Thus \(y\in P_S\), the
edge \(xy\) is directed along \(e_S\), and the simple straight suffix from
\(y\) to \(v\) continues in the same direction.  Therefore
\[
  S\subseteq\operatorname{Ext}_{e_S}(P_S).
\]

Conversely, start at \(y\in P_S\) and follow any grid ray in direction
\(e_S\).  Its \(\eta_m\)-state distance remains \(m\), so every vertex on
the ray has label \(m\) or \(m-1\).  If the ray first returned to label
\(m-1\), the preceding label-\(m\) vertices would lie in \(S\), and the
new lower vertex would lie in the unique predecessor \(A\).  This would
be an \(S\)-to-\(A\) cross edge directed along \(e_S\), so its
low-to-high direction would be \(-e_S\).  It cannot belong to the unique
straight band in (c), and would instead create a second \(A\)--\(S\)
Reeb band, contradicting \cref{cor:hf-reeb-tree}.  The ray therefore
remains in label
\(m\) and, being connected to \(P_S\subseteq S\), remains in \(S\).
This proves the reverse inclusion and (d).

Finally, compress any grid path from \(W_S\setminus P_S\) to
\(V(G)\setminus W_S\) into its sequence of sector runs.  In the rooted
tree, the unique edge leaving the descendant subtree \(\mathcal T_S\) is
\(SA\).  By (c), the \(S\)-endpoint of every grid edge realizing that
sector edge belongs to \(P_S\).  The original path therefore meets
\(P_S\), proving (e).
\end{proof}

\begin{remark}[Why ``gate'' is the exact universal statement]
\label{rem:hf-gate-not-cut}
One should not strengthen \cref{thm:hf-osp}(e) to the unqualified claim
that \(G-P_S\) is disconnected.  On one unit square, with a horizontal
stub at a lower corner, the root sector is the lower edge and the unique
positive sector is the upper edge.  Here \(P_S=S\), and deleting it leaves
the connected lower edge.  The gate statement is valid, and it is the
form needed below; a genuine vertex cut follows whenever the descendant
side contains a vertex outside its baseline.
\end{remark}

\begin{figure}[t]
\centering
\begin{tikzpicture}[scale=0.8,dot/.style={circle,fill,inner sep=1.0pt}]
\begin{scope}
  \fill[green!15] (0,0) rectangle (5,1);
  \fill[orange!18] (1,1) rectangle (2,3);
  \fill[orange!18] (3,1) rectangle (4,3);
  \draw[gray!50] (0,0) grid (5,1);
  \draw[gray!50] (1,1) grid (2,3);
  \draw[gray!50] (3,1) grid (4,3);
  \draw[thick] (0,0) -- (5,0) -- (5,1) -- (4,1) -- (4,3) -- (3,3)
    -- (3,1) -- (2,1) -- (2,3) -- (1,3) -- (1,1) -- (0,1) -- cycle;
  \draw[ultra thick,blue!65!black] (0,1) -- (5,1);
  \draw[->,very thick,blue!65!black] (0,1) -- (-0.6,1);
  \node[blue!65!black,font=\scriptsize,above] at (-0.4,1.05) {$p$};
  \draw[ultra thick,red!70!black] (1,2) -- (2,2);
  \draw[ultra thick,red!70!black] (3,2) -- (4,2);
  \node[red!70!black,font=\scriptsize,left] at (0.95,2.3) {$P_S$};
  \draw[->,red!70!black] (1.5,2.1) -- (1.5,2.8);
  \draw[->,red!70!black] (3.5,2.1) -- (3.5,2.8);
  \draw[->,green!45!black] (2.5,0.85) -- (2.5,0.15);
  \node[font=\tiny] at (5.35,1) {$0$};
  \node[font=\tiny] at (5.35,0) {$1$};
  \node[font=\tiny] at (2.35,3.2) {$1$};
  \node[font=\tiny] at (4.35,3.2) {$1$};
\end{scope}
\begin{scope}[xshift=8.4cm,yshift=0.9cm]
  \node[dot,label={[font=\scriptsize]left:$S_0$}] (r) at (0,0) {};
  \node[dot,label={[font=\tiny]above left:bottom}] (u) at (-0.9,1.1) {};
  \node[dot,label={[font=\tiny]above:tooth}] (v) at (0,1.4) {};
  \node[dot,label={[font=\tiny]above right:tooth}] (w) at (0.9,1.1) {};
  \draw (r)--(u); \draw (r)--(v); \draw (r)--(w);
  \node[font=\scriptsize] at (0,-0.8) {$\GammaP\cong K_{1,3}$};
\end{scope}
\end{tikzpicture}
\caption{\Cref{thm:hf-osp} on a comb-shaped cubical disk with
\(p=((0,1),\HH)\).  The root sector (thick blue row) has level \(0\);
the bottom row and the two teeth are the level-one sectors.  Each
positive sector is exactly the one-sided extrusion of its straight
baseline \(P_S\) (red) in the arrowed direction, and \(P_S\) gates its
descendant side.  With \(m\) teeth the sector graph is \(K_{1,m}\):
the family used in \cref{thm:comb-tabs} to rule out
bounded suppressed cores.}
\label{fig:extrusion}
\end{figure}
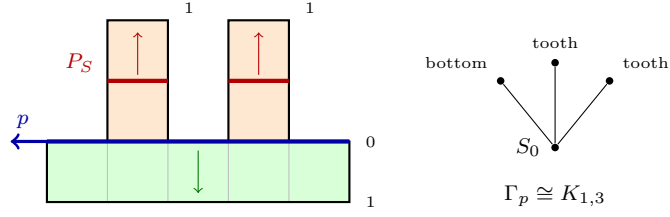

\section{Two-port common refinements}
\label{sec:joint}

For ports \(\mathcal P=(p_1,\ldots,p_r)\), set
\[
  B_{\mathcal P}(v)
  =
  \bigl(b^+_{p_1}(v),\ldots,b^+_{p_r}(v)\bigr).
\]
A \emph{common-refinement sector} is a connected component of
\(G[B_{\mathcal P}^{-1}(z)]\) for some
\(z\in\mathbb Z_{\ge0}^r\).  Let \(\Gamma_{\mathcal P}\) be the simple
graph whose vertices are these sectors and whose edges record grid
adjacency.  Equivalently, this is the final partition obtained by refining
with the ports one at a time; hence it is independent of their order.

\begin{observation}[Planarity and the one-port boundary case]
\label{obs:joint-basic}
The graph \(\Gamma_{\mathcal P}\) is planar.  Moreover,
\[
  \operatorname{tw}(\Gamma_p)\le h+1
\]
for one port \(p\).
\end{observation}

\begin{proof}
Contract a spanning tree inside every common-refinement sector and then
delete loops and merge parallel edges.  This obtains
\(\Gamma_{\mathcal P}\) as a minor of the planar graph \(G\).  For one
port, \cref{cor:algorithm} gives a feedback vertex set \(D\) of size at
most \(h\).  Put \(F=\Gamma_p-D\).  If \(F\) is nonempty, take width-one
tree decompositions of its components, join those decomposition trees
arbitrarily, and add \(D\) to every bag.  If \(F\) is empty, use the
single bag \(D\).  Thus every bag has size at most
\(|D|+2\le h+2\), so
\(\operatorname{tw}(\Gamma_p)\le h+1\).
\end{proof}

The simple-polygon theorem of
\cite[full version, Theorem~16]{deligkas2026polygon}
additionally states
\(\operatorname{tw}(\Gamma_{\mathcal P})\le 7^{r-1}\) when \(h=0\),
under its discretization and port conventions.  The next two subsections
explain why neither that theorem nor a possible hole-aware extension
should be rephrased as a bound on the number of sectors or on cycle
rank; \cref{subsec:two-port-sharp} then determines the filled-disk
constant exactly.  The broader edge-discretized scope is deferred to the
discussion in \cref{sec:edge-discussion}.

\subsection{Corridor suppression does not give a bounded core}
\label{subsec:comb}

Suppressing a degree-two vertex with two distinct neighbors means replacing
its two incident edges by one edge.  A \emph{topological core} retains
parallel edges created by different corridors, while a \emph{simple core}
merges them.  The distinction matters for cycle rank, but not for the two
disk constructions below.

\begin{theorem}[Unbounded suppressed size and branching on a disk]
\label{thm:comb-tabs}
There is no function of \(r+h\) that bounds either
\begin{enumerate}[label=(\roman*),leftmargin=2.2em]
  \item the number of vertices remaining after degree-two suppression, or
  \item the number of remaining vertices of degree at least three
\end{enumerate}
in \(\Gamma_{\mathcal P}\).  Both failures occur for \(r=1\) and \(h=0\),
for the augmented metric on pure cubical domains.
\end{theorem}

\begin{proof}
For (i), take a one-cell-high horizontal base rectangle, put a horizontal
port \(p\) at its upper-left corner, and attach \(m\) pairwise separated
one-cell-wide vertical rectangular teeth along its upper boundary.  The
entire upper boundary of the base has label zero.  All new vertices of
each tooth have label one; different teeth lie in different connected
components of that level.  Thus \(\Gamma_p\) contains a
\(K_{1,m}\) whose centre and leaves are not removed by degree-two
suppression.  Each tooth is attached along one boundary edge and has no
other contact, so the resulting pure cubical domain is a disk.

For (ii), we record an attachment invariant.  Suppose a level-\(i\)
sector \(S_i\) contains a boundary unit edge \(e\) of axis
\(\alpha_i\), and at both endpoints of \(e\) the orientation-state
distances satisfy
\[
  \delta_{\alpha_i}=i,
  \qquad
  \delta_{\bar\alpha_i}=i+1.
\]
Here \(\bar\alpha_i\) denotes the other axis.  Glue outside \(e\) a
one-cell-wide rectangle whose long axis is \(\bar\alpha_i\), with no
other contact with the old domain.  First observe that this attachment
does not decrease any old orientation-state distance.  Every state walk
that enters the new rectangle and later returns to the old domain can be
replaced by a walk supported on \(e\), using at most one switch at each
endpoint and no greater total weight.

Every edge from the old domain to a new vertex has axis
\(\bar\alpha_i\).  Hence reaching a new vertex in state
\(\bar\alpha_i\) costs at least \(i+1\), and reaching it in state
\(\alpha_i\) requires one further switch and costs at least \(i+2\).
Both bounds are attained by following a long side of the rectangle.
The orientation-state representation of \cref{lem:state} therefore gives
\[
  \delta_{\bar\alpha_i}=i+1,
  \qquad
  \delta_{\alpha_i}=i+2
\]
on the new long sides.  Because the rectangle is one cell wide, all its
new vertices lie on those two sides, which are joined by transverse grid
edges.  They form one level-\(i+1\) sector and provide new boundary edges
satisfying the invariant with the axes exchanged.

Start by attaching a sufficiently long vertical continuation \(S_1\) to
a label-zero top edge of the base.  We place the main continuations
monotonically, alternating upward and rightward.  Before branching from
\(S_i\), let \(U_i\) denote the union of the base, the earlier
continuations \(S_1,\ldots,S_{i-1}\), and the earlier leaves
\(L_1,\ldots,L_{i-1}\).  Maintain the following exposed-terminal
invariant: if \(i\) is odd, \(S_i\) is vertical and its top eight cells
lie strictly above \(\max_y U_i\); if \(i\) is even, \(S_i\) is
horizontal and its rightmost eight cells lie strictly to the right of
\(\max_x U_i\).

Suppose first that \(i\) is odd.  On the right long side of the terminal
eight-cell portion of \(S_i\), choose two vertical unit edges
\(e_i^L,e_i^C\) separated by at least two unused unit edges.  Attach a
\(2\times1\) leaf \(L_i\) to \(e_i^L\) and a one-cell-high horizontal
continuation \(S_{i+1}\) to \(e_i^C\), both extending to the right.
Their closed horizontal strips are disjoint, and these strips lie above
the bounding box of \(U_i\); hence the new rectangles have no contact with
one another or with the old domain except at their prescribed attachment
edges.  Choose the length of \(S_{i+1}\) so that its rightmost eight cells
lie strictly to the right of the bounding box of
\(U_i\cup S_i\cup L_i\).

If \(i\) is even, use the symmetric construction: choose two separated
horizontal unit edges on the top long side of the terminal portion of
\(S_i\), attach a \(1\times2\) leaf and a one-cell-wide vertical
continuation upward, and choose the latter long enough that its top eight
cells lie strictly above the bounding box of
\(U_i\cup S_i\cup L_i\).  Thus the exposed-terminal invariant is restored
with the axes exchanged.

The attachment invariant shows that \(L_i\) and \(S_{i+1}\) are distinct
level-\((i+1)\) sectors: their rectangles are disjoint and every path
between them meets the level-\(i\) sector \(S_i\).  Repeated gluing along
the prescribed boundary edges preserves the disk topology.  After
branching stages \(i=1,\ldots,N\), every
\(S_i\) has the three distinct neighbors
\(S_{i-1},L_i,S_{i+1}\).  Because \(h=0\), the one-port sector graph is a
tree; suppressing degree-two vertices therefore creates no parallel
identification among these branches.  Hence all \(S_i\) survive either
suppression convention.
\end{proof}

\subsection{Two fixed ports have unbounded joint cycle rank}
\label{subsec:stair}

The preceding trees have treewidth one.  The following family rules out a
stronger but still tempting argument: one cannot iterate the single-port
feedback-set bound and hope that the common refinement has cycle rank
bounded by the number of ports and holes.

For \(t\ge3\), let \(X_t\) be the union of the unit squares
\begin{equation}
  \{\Qcell{0}{0}\}
  \cup
  \{\Qcell{j}{j-1},\Qcell{j}{j}:1\le j\le t\}.
\label{eq:stair-cells}
\end{equation}
These squares form an alternating edge-glued staircase and hence a pure
cubical disk.  Fix, independently of \(t\), the two ports
\begin{equation}
  a=((0,1),\HH),
  \qquad
  c=((2,3),\HH),
\label{eq:stair-ports}
\end{equation}
and write \(Q_t=\Gamma_{(a,c)}\).  For a finite connected graph \(H\),
write \(B(H)=|\{v\in V(H):\deg_H(v)\ge3\}|\).

\begin{theorem}[Exact two-port staircase]
\label{thm:staircase}
For every \(t\ge3\),
\begin{equation}
  \betti(Q_t)=\left\lfloor\frac{2t}{3}\right\rfloor+2,
  \qquad
  B(Q_t)=\left\lfloor\frac{2t+1}{3}\right\rfloor+3,
  \qquad
  \operatorname{tw}(Q_t)=2.
\label{eq:stair-counts}
\end{equation}
Consequently, the raw common-refinement cycle rank and branch count are
unbounded for the fixed parameters \(r=2,h=0\).
\end{theorem}

\begin{proof}
We give an arithmetic certificate for the quotient.  Temporarily add the
fixed auxiliary port \(b=((1,0),\VV)\).  For each extension
\(X_{j-1}\subset X_j\), \(j\ge4\), put
\[
  D_j=(j,j),\quad E_j=(j+1,j-1),\quad
  F_j=(j+1,j),\quad G_j=(j,j+1).
\]
The formulas below are evaluated in \(X_j\); the no-shortcut argument below
shows that they remain valid in every later \(X_t\) containing the listed
vertices.  A direct state-graph calculation in \(X_4\) verifies the
certificate for
\[
 F_3,D_4,E_4,F_4,G_4,D_5,
\]
with three-coordinate scalar value
\[
  (f_a(F_3),f_b(F_3),f_c(F_3))=(2,2,0)
\]
and state pairs
\(((2,3),(3,2),(0,1))\) in port order \((a,b,c)\).  This is the induction
base.

In coordinate order \((a,b,c)\), the three augmented bend labels are
\begin{equation}
\begin{array}{c|ccc}
v&f_a(v)&f_b(v)&f_c(v)\\ \hline
D_j&
\left\lfloor(2j-1)/3\right\rfloor&
\left\lfloor2j/3\right\rfloor&
\left\lfloor(2j+1)/3\right\rfloor-2\\
E_j&
2\left\lfloor j/3\right\rfloor&
2\left\lfloor(j-1)/3\right\rfloor+1&
2\left\lfloor(j-2)/3\right\rfloor\\
F_j&
\left\lfloor2j/3\right\rfloor&
\left\lfloor(2j+1)/3\right\rfloor&
\left\lfloor(2j+2)/3\right\rfloor-2\\
G_j&
2\left\lfloor(j-1)/3\right\rfloor+1&
2\left\lfloor(j+1)/3\right\rfloor&
2\left\lfloor j/3\right\rfloor-1 .
\end{array}
\label{tab:stair-labels}
\end{equation}
These formulas follow by induction from the two-state shortest-path
recurrence encoded by \(\widehat G_p\) in \cref{sec:state}.  To make the
induction check explicit,
the minimizing axes for \(j\equiv0,1,2\pmod3\), respectively, are
\begin{equation}
\begin{array}{c|ccc}
 &0&1&2\\ \hline
D_j&(\VV,\VV,\HH)&(\HH,\VV,\VV)&(\VV,\HH,\VV)\\
E_j&(\HH,\HH,\HH)&(\HH,\HH,\HH)&(\HH,\HH,\HH)\\
F_j&(\HH,\VV,\HH)&(\HH,\HH,\VV)&(\VV,\HH,\HH)\\
G_j&(\VV,\VV,\VV)&(\VV,\VV,\VV)&(\VV,\VV,\VV).
\end{array}
\label{tab:stair-axes}
\end{equation}
For each coordinate the other state has value one larger.  For a port
\(q\in\{a,b,c\}\) and a non-source vertex \(v\), the state distances obey
\[
\begin{aligned}
 \delta^q_{\HH}(v)
 &=\min\!\left\{1+\delta^q_{\VV}(v),
     \min_{uv\ {\rm horizontal}}\delta^q_{\HH}(u)\right\},\\
 \delta^q_{\VV}(v)
 &=\min\!\left\{1+\delta^q_{\HH}(v),
     \min_{uv\ {\rm vertical}}\delta^q_{\VV}(u)\right\}.
\end{aligned}
\]
The two-square lobe \(X_j\setminus X_{j-1}\) meets the old staircase
exactly in the vertical edge \(F_{j-1}D_j\).  Any old-to-old state walk
that enters the lobe can be replaced by a walk supported on this edge,
using at most one switch at each endpoint and no greater weight.  Thus old
state distances do not decrease.  Substituting the old boundary values
into the two displayed recurrences for the four new vertices, separately
for \(j\bmod3=0,1,2\), gives the scalar and preferred-axis tables and
completes the induction.

The six frontier edges added at step \(j\) are
\[
 F_{j-1}E_j,\ E_jF_j,\ D_jF_j,\ D_jG_j,\
 G_jD_{j+1},\ F_jD_{j+1}.
\]
Substituting \cref{tab:stair-labels} shows that exactly one of
\((a,b,c)\) changes on each such edge.  For residues \(0,1,2\), the six
changing coordinates are, respectively,
\begin{equation}
\begin{array}{c|cccccc}
j\bmod3&
F_{j-1}E_j&E_jF_j&D_jF_j&D_jG_j&G_jD_{j+1}&F_jD_{j+1}\\ \hline
0&a&b&a&c&a&c\\
1&b&c&b&a&b&a\\
2&c&a&c&b&c&b .
\end{array}
\label{tab:stair-colors}
\end{equation}
Dropping the \(b\)-coordinate identifies exactly the connected components
of the subgraph of \(\Gamma_{(a,b,c)}\) formed by adjacencies whose
endpoints have the same \((a,c)\)-label.  Equivalently, one contracts all
\(b\)-only adjacencies and then deletes loops and merges parallel edges.
Indeed, a constant-\((a,c)\) path in the original grid becomes a walk of
\(b\)-only quotient edges, and every such quotient walk lifts to a
constant-\((a,c)\) grid path.

Applying this fact to the six frontier edges gives three local quotients.
If \(j\equiv0\pmod3\), then
\([F_{j-1}]=[D_j]=P\) and \([E_j]=[F_j]=X\); with
\(Y=[G_j]\) and \(Z=[D_{j+1}]\), the new edges are
\[
 PX,\quad PY,\quad XZ,\quad YZ.
\]
If \(j\equiv1\pmod3\), then
\[
 [E_j]=[F_{j-1}],\qquad
 [F_j]=[D_j],\qquad
 [G_j]=[D_{j+1}],
\]
and deleting loops and merging edges parallel to the old attachment leaves
one new vertex and one new edge.  If \(j\equiv2\pmod3\), then
\([G_j]=[D_j]\) and \([F_j]=[D_{j+1}]\); a three-edge path
\[
 [F_{j-1}]-[E_j]-[F_j]-[D_j]
\]
is added in parallel with the old edge
\([F_{j-1}][D_j]\).  Tracking the endpoint degrees from the base
\(t=3\) gives the increments
\begin{equation}
\begin{array}{c|rrrr}
j\bmod3&
\Delta|V|&\Delta|E|&\Delta\betti&\Delta B\\ \hline
0&3&4&1&1\\
1&1&1&0&1\\
2&2&3&1&0 .
\end{array}
\label{tab:stair-recurrence}
\end{equation}
In the first case the old vertex \(P\) changes from degree two to degree
four.  In the second, the class containing \(D_j\) changes from degree two
to degree three.  In the third, the old endpoints are respectively already
a branch vertex and a leaf, and become degree four and degree two; no new
branch vertex is created.  These assertions hold in the \(t=3\) base and
are preserved by the same three local updates, proving the
\(\Delta B\) column together with the other three columns.
The directly checked base \(t=3\) is
\((|V|,|E|,\betti,B)=(11,14,4,5)\).  Summing the three-periodic recurrence
gives the first two identities in \cref{eq:stair-counts}.

Finally, the grid graph \(X_t^{(1)}\) is obtained by repeatedly taking
two-clique-sums of four-cycles along the shared square edges, so it has
treewidth at most two.  Since \(Q_t\) is a minor of this graph,
\(\operatorname{tw}(Q_t)\le2\).  Its positive cycle rank shows that it is
not a forest, hence its treewidth is at least two.
\end{proof}

\begin{corollary}[Unbounded joint feedback number]
\label{cor:stair-fvs}
If \(\tau_{\rm FVS}(Q_t)\) denotes the minimum size of a feedback vertex
set, then
\[
 \tau_{\rm FVS}(Q_t)
 \ge\left\lfloor\frac{t-2}{3}\right\rfloor.
\]
\end{corollary}

\begin{proof}
For \(j=3q+2\), with
\(1\le q\le\lfloor(t-2)/3\rfloor\), the four sectors containing
\[
 F_{j-1},\qquad E_j,\qquad
 F_j\sim D_{j+1},\qquad D_j\sim G_j
\]
have respective two-port labels
\[
 (2q,2q-1),\quad(2q,2q),\quad
 (2q+1,2q),\quad(2q+1,2q-1)
\]
and form the four-cycle described by the
\(j\equiv2\pmod3\) local quotient.  Different values of \(q\) use
disjoint label sets, so these cycles are vertex-disjoint.  Every feedback
vertex set must meet each of them.
\end{proof}

\begin{remark}[What the staircase does and does not refute]
\label{rem:stair-scope}
If degree-two suppression retains parallel corridors, it preserves the
\(\Theta(t)\) cycle rank in \cref{thm:staircase}.  If it subsequently
merges parallel edges, the diamond chain may collapse to a path; the comb
and tab constructions in \cref{thm:comb-tabs} independently handle simple
core size and branching.  None of these families has unbounded treewidth.
Together with \cref{cor:stair-fvs}, they invalidate bounded-size and
bounded-joint-feedback-set proof strategies, not the conjecture
\(\operatorname{tw}(\Gamma_{\mathcal P})\le f(r,h)\).
\end{remark}

\subsection{The exact constant for two ports on cubical disks}
\label{subsec:two-port-sharp}

We now specialize to the filled class for which the single-port geometry
of \cref{thm:hf-osp} is available.  Let \(\mathcal D_\square\) be the class
of finite pure planar cubical complexes whose realizations are closed
disks.  For \(X\in\mathcal D_\square\), \(G=X^{(1)}\), and two augmented
ports \(p,q\), define
\[
  c_2^\square
  =
  \sup_{X\in\mathcal D_\square,\ p,q}
  \operatorname{tw}(\Gamma_{p,q}).
\]
This constant is distinct from a supremum over more general
edge-discretized polygons with lower-dimensional appendages.

We first record the lower-bound example (\cref{fig:twobytwo}); the
remainder of the subsection proves the matching upper bound.

\begin{proposition}[A \(2\times2\) cubical disk gives the \(3\times3\)
grid]
\label{prop:two-port-grid}
Let \(X=[0,2]\times[0,2]\) with its unit cubical subdivision, and take
\[
 p=((0,1),\HH),
 \qquad
 q=((1,0),\VV).
\]
For the augmented port metric,
\[
 \Gamma_{p,q}\cong P_3\square P_3,
 \qquad
 \operatorname{tw}(\Gamma_{p,q})=3.
\]
\end{proposition}

\begin{proof}
At every \((x,y)\in\{0,1,2\}^2\), the two augmented bend distances are
\[
 b_p^+(x,y)=\mathbf1_{\{y\ne1\}},
 \qquad
 b_q^+(x,y)=\mathbf1_{\{x\ne1\}}.
\]
Indeed, the middle row is reachable from the horizontal virtual stub
without a switch, while every other row requires exactly one switch; the
second formula is the vertical analogue.  Every horizontal grid edge
changes the second indicator and every vertical grid edge changes the
first.  Consequently, no two adjacent grid vertices have the same joint
label.  Equal labels at nonadjacent vertices lie in different connected
components, so all nine common-refinement sectors are singletons and the
sector adjacency graph is exactly \(P_3\square P_3\).

For the lower bound, use the four connected, pairwise disjoint branch sets
\[
\begin{aligned}
A&=\{(1,1)\},\\
B&=\{(1,2),(0,2),(0,1),(0,0)\},\\
C&=\{(2,2),(2,1)\},\\
D&=\{(1,0),(2,0)\}.
\end{aligned}
\]
The centre \(A\) is adjacent to each other branch set, and the edges
\[
 (1,2)(2,2),\qquad (2,1)(2,0),\qquad (0,0)(1,0)
\]
witness \(BC,CD,DB\).  Thus these branch sets form a \(K_4\)-minor.
Conversely, eliminating the vertices in the order
\[
(0,0),(0,2),(2,0),(2,2),(0,1),(1,0),(1,1),(1,2),(2,1)
\]
has filled later-neighborhood sizes
\[
2,2,2,2,3,3,2,1,0.
\]
It therefore gives a width-three chordal completion.
\end{proof}

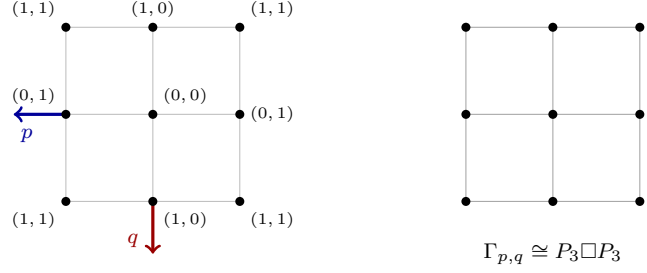
\begin{figure}[t]
\centering
\begin{tikzpicture}[scale=1.15,dot/.style={circle,fill,inner sep=1.2pt}]
\begin{scope}
  \draw[gray!50] (0,0) grid (2,2);
  \draw[->,very thick,blue!60!black] (0,1) -- (-0.6,1);
  \node[blue!60!black,font=\scriptsize,below] at (-0.45,0.95) {$p$};
  \draw[->,very thick,red!60!black] (1,0) -- (1,-0.6);
  \node[red!60!black,font=\scriptsize,left] at (0.95,-0.45) {$q$};
  \foreach \x/\y/\l/\pos in {
    0/0/{(1,1)}/{below left},
    1/0/{(1,0)}/{below right},
    2/0/{(1,1)}/{below right},
    0/1/{(0,1)}/{above left},
    1/1/{(0,0)}/{above right},
    2/1/{(0,1)}/{right},
    0/2/{(1,1)}/{above left},
    1/2/{(1,0)}/{above},
    2/2/{(1,1)}/{above right}}
  {
    \node[dot] at (\x,\y) {};
    \node[font=\tiny,\pos] at (\x,\y) {$\l$};
  }
\end{scope}
\begin{scope}[xshift=4.6cm]
  \draw[gray!70] (0,0) grid (2,2);
  \foreach \x in {0,1,2} \foreach \y in {0,1,2}
    \node[dot] at (\x,\y) {};
  \node[font=\scriptsize] at (1,-0.6)
    {$\Gamma_{p,q}\cong P_3\square P_3$};
\end{scope}
\end{tikzpicture}
\caption{\Cref{prop:two-port-grid}: the \(2\times2\) cubical disk with
ports \(p=((0,1),\HH)\) and \(q=((1,0),\VV)\).  Every vertex carries its
joint label \((b^+_p,b^+_q)\); no two adjacent vertices agree, so all
nine joint sectors are singletons and the common refinement is the
\(3\times3\) grid, of treewidth three.}
\label{fig:twobytwo}
\end{figure}

We first record an elementary rigidity fact about hyperplanes of planar
grid complexes, which the convexity argument below uses.

\begin{lemma}[Hyperplanes of a planar grid complex are straight strips]
\label{lem:sharp-hyperplane-straight}
Let \(X\) be a finite pure planar cubical domain, and consider the
equivalence relation on grid edges generated by declaring the two
opposite sides of every unit square of \(X\) equivalent.  Every
equivalence class \(\mathfrak h\) consists of edges of one axis lying in
one unit coordinate slab: after exchanging axes, say vertical edges
\(\{x\}\times[y_0,y_0+1]\) for \(x\) in a set \(I\subseteq\Z\).  The set
\(I\) is an integer interval, the squares of \(X\) joining consecutive
edges of \(\mathfrak h\) form the contiguous straight strip
\([\min I,\max I]\times[y_0,y_0+1]\subseteq X\), and no square of \(X\)
outside this strip contains two edges of \(\mathfrak h\).  In
particular, the hyperplane dual to \(\mathfrak h\) neither bends nor
branches, and its carrier is that straight strip.
\end{lemma}

\begin{proof}
The generating relation pairs a vertical edge
\(\{x\}\times[y_0,y_0+1]\) only with the opposite side
\(\{x\pm1\}\times[y_0,y_0+1]\) of a square \(Q_{x,y_0}\) or
\(Q_{x-1,y_0}\) of \(X\): opposite sides of a unit square are parallel
and lie in the same unit slab.  Hence the axis and the slab are
invariants of the class, and within the slab one generating step changes
the \(x\)-coordinate by exactly one, through a square of \(X\) in that
slab.  The class is therefore a connected set of integers under steps
through present squares, so \(I\) is an interval and the connecting
squares form the stated contiguous strip.  A unit square has exactly one
pair of vertical sides, so it contains at most two edges of
\(\mathfrak h\), and precisely the strip squares contain two.
\end{proof}

\begin{lemma}[Directional extrusions are convex]
\label{lem:sharp-extrusion-convex}
Let \(X\in\mathcal D_\square\), let \(P\) be a straight grid path, and let
\(e\) be a cardinal direction orthogonal to \(P\).  Then
\(\operatorname{Ext}_e(P)\) is graph-convex in \(G=X^{(1)}\).
\end{lemma}

\begin{proof}
The cubical disk \(X\) is simply connected, and every vertex link is a
subgraph of \(C_4\), hence flag.  Thus \(X\) is a CAT(0) cube complex.
Cubical hyperplanes separate its one-skeleton into graph-convex
halfspaces \cite{farley2009sageev}.

Put \(H=\operatorname{Ext}_e(P)\).  We show that every boundary grid edge
\(uv\), with \(u\in H\) and \(v\notin H\), is dual to a hyperplane having
all of \(H\) on the \(u\)-side.  If \(uv\) is parallel to \(e\), it cannot
leave \(H\) in the forward direction, since the same fibre would reach
\(v\).  It therefore points backward.  Writing \(u=x+re\) for
\(x\in P\), a backward neighbor could fail to lie in \(H\) only when
\(r=0\); hence \(u\in P\).  No edge internal to the forward extrusion is
dual to this behind-the-baseline hyperplane.  Since \(H\) is connected,
it lies entirely in the \(u\)-side halfspace.

Suppose instead that \(uv\perp e\).  If its dual hyperplane also crossed
an edge internal to \(H\), then its carrier would connect the two edges by
a straight strip of squares (\cref{lem:sharp-hyperplane-straight}).
After rotating the picture, take
\(e=+\mathsf y\) and \(P\) horizontal at height \(y_0\).  An internal
dual edge has both endpoints on the \(+\mathsf y\)-fibres from two
adjacent vertices of \(P\).  The boundary edge \(uv\) and that internal
edge are both at heights at least \(y_0\).  The straight carrier contains
every intervening square at their fixed pair of column coordinates, so
its vertical sides extend both fibres square by square to the height of
\(uv\).  Both endpoints of \(uv\) would then lie in \(H\), a contradiction.
Since \(H\) is connected, it cannot have vertices on both sides of this
hyperplane without an internal edge crossing it.

For each \(z\notin H\), take a distance-minimizing graph path from \(z\)
to the set \(H\), and use the hyperplane dual to its last edge.  That edge
is a boundary edge of \(H\).  A graph geodesic in a CAT(0) cube complex
crosses a hyperplane at most once, so the path did not cross this
hyperplane earlier; consequently it separates \(z\) from all of \(H\).
Thus \(H\) is an intersection of cubical halfspaces and is graph-convex.
\end{proof}

\begin{corollary}[Convex single-port sectors]
\label{cor:sharp-sector-convex}
Every augmented single-port sector in \(X\in\mathcal D_\square\) is
graph-convex.  Hence the intersection of a \(p\)-sector and a \(q\)-sector
is empty or connected.
\end{corollary}

\begin{proof}
The root sector is a straight path and is convex by the coordinate
distance lower bound.  Every other sector is a directional extrusion by
\cref{thm:hf-osp}, so \cref{lem:sharp-extrusion-convex} applies.
Intersections of graph-convex sets are graph-convex.
\end{proof}

Let \(T_p,T_q\) denote the two rooted single-port sector trees.  By the
last corollary, every joint sector is precisely one nonempty intersection
of a \(p\)-sector with a \(q\)-sector; no such pair splits into two joint
sectors.  Let
\[
  \phi_p:V(\Gamma_{p,q})\to V(T_p),
  \qquad
  \phi_q:V(\Gamma_{p,q})\to V(T_q)
\]
record the containing single-port sectors.

\begin{lemma}[A straight path meets at most three sectors]
\label{lem:sharp-three-sectors}
For one augmented port, every straight grid path meets at most three
single-port sectors.  The sectors met lie on a path of length at most two
in the rooted sector tree.
\end{lemma}

\begin{proof}
Convexity makes the intersection of each sector with the straight path an
interval, so no sector repeats.  Projecting consecutive intervals gives a
simple path in the sector tree.

Let \(z\) have minimum bend label \(m\) on the straight path.  Append either
side of that path to an optimal state path ending at \(z\).  At \(z\) this
requires at most one additional switch, so every label on the straight
path is \(m\) or \(m+1\).  By \cref{thm:hf-osp}, these are also the rooted
depths.  A simple path in a rooted tree using only two consecutive depths
has the form child--parent--child and has at most three vertices.
\end{proof}

\begin{lemma}[Fibre trees]
\label{lem:sharp-fibre-tree}
For every \(A\in V(T_p)\), the induced graph
\[
  H_A=\Gamma_{p,q}[\phi_p^{-1}(A)]
\]
is a tree.  The symmetric statement holds for every \(q\)-sector.
\end{lemma}

\begin{proof}
The map \(\phi_q\) is injective on \(\phi_p^{-1}(A)\), since each nonempty
intersection \(A\cap B\) is connected.  A path in the connected graph
\(G[A]\), projected through its successive \(q\)-sectors, proves that
\(H_A\) is connected.  Every edge of \(H_A\) maps to an edge of \(T_q\),
so injectivity identifies \(H_A\) with a connected subgraph of the tree
\(T_q\).
\end{proof}

\begin{lemma}[The interface between adjacent old sectors]
\label{lem:sharp-convex-bridge}
If \(A,A'\) are adjacent in \(T_p\), all grid edges between them form an
order-preserving perfect matching between two equal-length parallel
straight paths \(Q_A\subseteq A\) and \(Q_{A'}\subseteq A'\).
\end{lemma}

\begin{proof}
Root \(T_p\) at its level-zero sector.  Two sectors adjacent in the
rooted tree \(T_p\) are in a parent--child relation; say \(A'\) is the
child, at level \(m\).  \cref{thm:hf-osp}(c) states that the
\(A\)--\(A'\) grid edges are precisely the order-preserving perfect
matching
\[
  x(x+e_{A'}),\qquad x\in\overline P_{A'},
\]
between the equal-length parallel straight paths
\(\overline P_{A'}\subseteq A\) and
\(P_{A'}=\overline P_{A'}+e_{A'}\subseteq A'\).  Setting
\(Q_A=\overline P_{A'}\) and \(Q_{A'}=P_{A'}\) gives the statement.
\end{proof}

For an edge \(e=AA'\in E(T_p)\), let
\(\mathcal A_{A,e}\subseteq V(H_A)\) be the joint sectors containing an
\(A\)-endpoint of an \(A\)--\(A'\) grid edge, and define
\(\mathcal A_{A',e}\) symmetrically.

\begin{corollary}[Three-vertex adhesions]
\label{cor:sharp-adhesion}
Each set \(\mathcal A_{A,e}\) has size at most three and is contained in
a path of length at most two in \(H_A\).
\end{corollary}

\begin{proof}
Apply \cref{lem:sharp-three-sectors} for port \(q\) to the straight shore
\(Q_A\) of \cref{lem:sharp-convex-bridge}.  Convexity makes the successive
\(q\)-sectors a path in \(T_q\); the injective identification in
\cref{lem:sharp-fibre-tree} transfers that path to \(H_A\).
\end{proof}

We use two elementary width lemmas.

\begin{lemma}[Planar distance-two completion of a tree]
\label{lem:sharp-planar-tree}
Let \(T\) be a tree and let \(L\) be a planar graph obtained from \(T\) by
adding edges whose endpoints have distance two in \(T\).  Then
\(\operatorname{tw}(L)\leq3\).
\end{lemma}

\begin{proof}
For \(x\in V(T)\), let \(L_x=L[\{x\}\cup N_T(x)]\), assigning each added
edge to its unique middle vertex.  The vertex \(x\) is universal in the
planar graph \(L_x\).  Deleting it leaves an outerplanar graph, and hence
\(\operatorname{tw}(L_x)\leq3\).

Take width-three decompositions of all \(L_x\).  For every tree edge
\(xy\), join a bag containing the clique \(\{x,y\}\) on the \(L_x\) side
to such a bag on the \(L_y\) side.  These joins follow \(T\).  A vertex
\(v\) occurs in the block \(L_v\) and in the blocks indexed by its
neighbors, which form a star, so running intersection holds.
\end{proof}

\begin{lemma}[A three-by-three monotone interface]
\label{lem:sharp-monotone-interface}
Let \(A=\{a_1,\ldots,a_r\}\) and \(B=\{b_1,\ldots,b_s\}\) be ordered
sets, where \(r,s\leq3\).  Make both sets cliques, and add cross edges
whose index pairs form a subset of a monotone lattice path.  The
resulting graph has treewidth at most three.
\end{lemma}

\begin{proof}
Deleting cross edges does not increase treewidth, so complete the subset
to a maximal monotone path.  For \(r=s=3\), its step word from
\((1,1)\) to \((3,3)\) uses
\(\mathsf A=(1,0)\), \(\mathsf B=(0,1)\), and
\(\mathsf X=(1,1)\), and is one of the following thirteen words:
\[
\begin{array}{c|c@{\qquad}c|c}
\text{word}&\text{\(A\)-elimination order}&
\text{word}&\text{\(A\)-elimination order}\\ \hline
\mathsf{AABB}&123&\mathsf{ABAB}&123\\
\mathsf{ABBA}&132&\mathsf{ABX}&123\\
\mathsf{AXB}&123&\mathsf{BAAB}&213\\
\mathsf{BABA}&321&\mathsf{BAX}&213\\
\mathsf{BBAA}&231&\mathsf{BXA}&231\\
\mathsf{XAB}&123&\mathsf{XBA}&132\\
\mathsf{XX}&123&&
\end{array}
\]
A direct elimination check shows that every eliminated \(a_i\) has at
most three later neighbors after fill; the three \(B\)-vertices then form
a triangle.  If \(r<3\) or \(s<3\), add dummy vertices and extend the path
monotonically to a \(3\times3\) instance.  The original graph is a
subgraph of that instance.
\end{proof}

\begin{theorem}[Sharp two-port upper bound]
\label{thm:sharp-two-port-upper}
For every \(X\in\mathcal D_\square\) and every two augmented ports \(p,q\),
\[
  \operatorname{tw}(\Gamma_{p,q})\leq3.
\]
\end{theorem}

\begin{proof}
For every \(A\in V(T_p)\), start with the fibre tree \(H_A\).  For each
incident edge \(e\in E(T_p)\), complete
\(\mathcal A_{A,e}\) to a clique; call the resulting local graph \(L_A\).
Every added edge has endpoints at distance two in \(H_A\), by
\cref{cor:sharp-adhesion}.

The graph \(L_A\) is planar.  To see this, in the planar joint-sector graph
contract, for every neighbor \(A'\) of \(A\), the connected fibre
\(H_{A'}\) to a hub, and delete all edges not incident with \(H_A\).
The hub neighborhood is exactly \(\mathcal A_{A,AA'}\) and has size at
most three.  Deleting each hub and completing its neighborhood is a planar
\(Y\)-\(\Delta\) operation.  Perform these replacements sequentially:
each step preserves planarity and does not change the degrees of the
remaining hubs.  The result is \(L_A\).  Thus
\cref{lem:sharp-planar-tree} gives
\[
  \operatorname{tw}(L_A)\leq3.
\]

For \(e=AA'\), form an interface block \(K_e\) on
\(\mathcal A_{A,e}\cup\mathcal A_{A',e}\), retaining the original cross
edges and completing each side to a clique.  Scan the two straight shores
of \cref{lem:sharp-convex-bridge} in their matching order.  On each shore,
the successive \(q\)-sector indices are nondecreasing along a path of at
most three sectors.  Hence the cross-edge index pairs form a monotone
lattice-path subset, and
\cref{lem:sharp-monotone-interface} gives
\(\operatorname{tw}(K_e)\leq3\).

Take width-three decompositions of every \(L_A\) and \(K_e\).  Each
adhesion is a clique of size at most three, so some bag on either side
contains it.  Join those bags according to the barycentric subdivision
\[
  L_A-K_{AA'}-L_{A'}
\]
of the tree \(T_p\).  A joint-sector vertex appears in its unique local
block and in precisely the incident interface blocks that contain it;
these blocks form a star.  Running intersection therefore holds.  Local
and cross-interface edges are all covered, and no bag has more than four
vertices.
\end{proof}

\begin{corollary}[The sharp filled two-port constant]
\label{cor:sharp-two-port-constant}
\[
  \boxed{c_2^\square=3.}
\]
\end{corollary}

\begin{proof}
The upper bound is \cref{thm:sharp-two-port-upper}.  The \(2\times2\)
example of \cref{prop:two-port-grid} belongs to
\(\mathcal D_\square\) and has refinement graph
\(P_3\square P_3\) of treewidth three.
\end{proof}

\begin{remark}[Scope of the equality]
\label{rem:sharp-two-port-scope}
The equality concerns filled cubical disks.  The coarse
\(3\leq c_2^{\mathrm{edge}}\leq11\) statement of
\cref{cor:two-port-eleven} has a different, edge-discretized source
scope and is deferred to the discussion in \cref{sec:edge-discussion}.  In particular, the equality above is
not silently extended to thin appendages, point pinches, holed domains, or
three or more ports.
\end{remark}

\section{Discussion: the edge-discretized source scope}
\label{sec:edge-discussion}

Every result so far is unconditional.  This section has a different
status, and we separate it deliberately: it records what can currently
be said in the broader edge-discretized scope of
\cite{deligkas2026polygon}, where the extrusion geometry of
\cref{thm:hf-osp} is not available and where our own machinery does not
apply because purity fails.  The outcome is an upper bound of eleven,
derived from two explicitly stated structural inputs of the full-version
refinement argument of \cite{deligkas2026polygon}; the bound is
conditional on those inputs holding under the operational augmented
convention.

The next statement isolates the graph theory needed for an upper bound.
It deliberately makes the geometric six-child input a hypothesis.

\begin{theorem}[Bridge/core refinement bound]
\label{thm:two-port-bridge-core}
Let \(\Gamma\) be a connected graph, let \(T\) be a tree, and let
\(\phi:V(\Gamma)\to V(T)\) satisfy
\[
 AB\in E(\Gamma)
 \quad\Longrightarrow\quad
 \phi(A)=\phi(B)
 \ \text{or}\
 \phi(A)\phi(B)\in E(T).
\]
Suppose \(N\subseteq V(\Gamma)\) has
\[
 |N\cap\phi^{-1}(S)|\le6
 \qquad(S\in V(T)),
\]
and every edge with an endpoint in \(V(\Gamma)\setminus N\) is a bridge
of \(\Gamma\).  Then
\[
 \operatorname{tw}(\Gamma)\le11.
\]
\end{theorem}

\begin{proof}
Put \(N_S=N\cap\phi^{-1}(S)\) and \(H=\Gamma[N]\).  Subdivide every edge
of \(T\) once to obtain a tree \(T^\star\).  Give an original tree node
\(S\) the bag
\[
 \beta(S)=N_S,
\]
and give the subdivision node corresponding to \(SS'\in E(T)\) the bag
\[
 \beta(e_{SS'})=N_S\cup N_{S'}.
\]
Every edge of \(H\) has both endpoints in one such bag.  A vertex in
\(N_S\) occurs at \(S\) and at precisely the subdivision nodes incident
with \(S\), which form a star.  Hence these bags decompose \(H\), and
\[
 |\beta(S)|\le6,\qquad
 |\beta(e_{SS'})|\le12.
\]

Delete all bridges of \(\Gamma\).  Every resulting component containing
an edge lies in \(H\), because every edge incident with a vertex outside
\(N\) was a bridge.  For each such component, take an independent copy of
the restriction of the preceding decomposition; for each singleton
component \(\{v\}\), take the singleton bag \(\{v\}\).  Contracting these
bridge-free components turns \(\Gamma\) into a tree.  For every bridge
\(uv\) between two components, add a bag \(\{u,v\}\) and attach it to a
bag containing \(u\) on one side and a bag containing \(v\) on the other.
The resulting host is obtained by replacing every vertex of the
bridge-block tree by its independent decomposition tree and every
block-tree edge by one bridge-bag path.  It is therefore a tree, although
not necessarily a literal subdivision of the block tree.  Vertex
occurrence sets remain connected, and the new bags have size at most two.
The maximum bag size is therefore twelve.
\end{proof}

\begin{corollary}[A conditional two-port bound in the source scope]
\label{cor:two-port-eleven}
Let \(p,q\) be two ports in a discretized simple polygon, under the
operational augmented interpretation of the sector construction in
\cite{deligkas2026polygon}.  Assume that the following two conclusions
of its full version hold in this interpretation for the port subsequence
\(p,q\):
\begin{enumerate}[label=(\roman*),leftmargin=2.2em]
  \item the one-port quotient \(\Gamma_p\) is a tree (Lemma~9); and
  \item with the Definition~13 designation of clean joint sectors, each
  \(p\)-sector contains at most six non-clean joint sectors, and every
  clean joint sector has exactly one neighbor in each component of its
  deletion (Lemma~14).
\end{enumerate}
Then
\[
 \operatorname{tw}(\Gamma_{p,q})\le11.
\]
Consequently, if \(c_2^{\mathrm{edge}}\) denotes the supremum of the
two-port augmented refinement treewidth over this edge-discretized class,
then under hypotheses \textup{(i)}--\textup{(ii)},
\[
 3\le c_2^{\mathrm{edge}}\le11.
\]
\end{corollary}

\begin{proof}
Let \(\phi\) map every joint sector to its containing \(p\)-sector; by
(i) the target
\(T=\Gamma_p\) is a tree.  A grid edge witnessing adjacency of two joint
sectors has endpoints in the same or adjacent \(p\)-sectors, so the
projection satisfies the edge condition of
\cref{thm:two-port-bridge-core}.

Let \(N\) be the non-clean sectors of (ii), so
\(|N\cap\phi^{-1}(S)|\le6\) for every \(p\)-sector \(S\).  By the second
part of (ii), every edge incident with a clean sector is a bridge:
otherwise a cycle through that edge would supply a second neighbor in
the same deletion component.  The theorem applies.  The lower bound is
\cref{prop:two-port-grid}.
\end{proof}

\begin{remark}[Why the hypotheses are stated explicitly]
\label{rem:two-port-transfer}
\Cref{cor:literal-c3} shows that under the literal reading of the formal
bend-distance definition in \cite{deligkas2026polygon}, the tree
statement (i) is already false on a unit square.  The operational
procedure of that work is exactly \cref{prop:layering}, so its geometric
proofs are naturally read in the augmented convention; moreover, the
published proof of its Lemma~14 refers only to the already incorporated
ports and the newly added port, so restricting to the subsequence
\(p,q\) does not use the full-terminal consecutive ordering.  We
nevertheless do not re-derive (i)--(ii) for the augmented metric on
edge-discretized domains here: the Reeb machinery of
\cref{sec:reeb,sec:bound} requires purity and does not directly apply to
graphs with dangling one-dimensional parts.  The bound eleven is
therefore conditional on this metric transfer, in contrast with the
unconditional filled-disk equality of
\cref{cor:sharp-two-port-constant}.
\end{remark}

\begin{remark}[Statement versus public proof arithmetic]
\label{rem:two-port-source-arithmetic}
The full-version Theorem~16 of \cite{deligkas2026polygon} \emph{states}
the stronger bound
\(\operatorname{tw}(\Gamma_{\mathcal P})\le7^{r-1}\), which would give
\(c_2^{\mathrm{edge}}\le7\).  We claim no counterexample to that
statement.  In the public
arXiv v1 proof, however, the displayed induction index, treewidth, and
maximum bag size are not aligned: replacing each old bag member by at
most six non-clean children multiplies bag size by six, and the displayed
splitting step adds a clean sector.  Starting from bag size two, that
arithmetic gives bag size thirteen, hence width twelve, conditional on
separately repairing the splitting order and attachment details.  The
bridge/core proof above does not use that splitting construction and is
why the weaker number eleven is derived here from the explicitly
isolated hypotheses of \cref{cor:two-port-eleven}.
\end{remark}

\section{Conclusion and open problems}
\label{sec:conclusion}

The endpoint convention is mathematically consequential.  Under a literal
last-actual-edge interpretation, bend distance has a non-affine
\((0,0,1,2)\) pattern on one square, and the literal sector graph of that
hole-free square is a triangle, so the tree theorem and the cycle bound
already fail there (\cref{cor:literal-c3}).  Under the explicitly repaired
virtual-stub convention---the convention implemented by the operational
bend-layering procedure---the orientation-state representation forbids
saddles inside unit squares and selects a canonical cubical extension.
Its Reeb multigraph retains the topology of open bend bands; the usual
sector graph is only its parallel-edge simplification.  For this augmented
metric on pure two-dimensional cubical domains we obtain
\(\betti(\GammaP)\le h\) and a linear-time feedback-set certificate of
size at most \(h\).  On a filled cubical disk, retaining Reeb-edge
multiplicity and using preferred-axis parity gives more: the Reeb
multigraph is a tree, every positive sector has a unique straight parent
interface, its label is its rooted depth, and it is exactly the one-sided
extrusion of that interface.  Convexity and the resulting three-vertex
adhesions then give the exact filled two-port constant
\(c_2^\square=3\), while a separate bridge/core argument gives the
conditional bound \(3\le c_2^{\mathrm{edge}}\le11\) in the broader
edge-discretized scope (\cref{rem:two-port-transfer}).

The same analysis also explains why the holes extension of coordinated
motion planning is not a formal substitution into the hole-free proof.
Parallel Reeb bands can be hidden by the simple sector graph, a one-hole
edge discretization can create two predecessors and a \(C_4\), and even a
tree quotient can contain a sector with no straight baseline.  Moreover,
bounded-size and feedback-set reformulations of the multi-port refinement
fail already on hole-free disks, while treewidth remains viable.  The
topological theorem is a structural first step and a falsification tool,
not a proof of fixed-parameter tractability.

\paragraph{What the cycle bound does not prove.}
Deleting at most \(h\) sector vertices destroys all cycles of a
single-port quotient, but the surviving sectors need not have straight
baselines (\cref{ex:nobaseline}), and a sector marked as one topological
defect for one port can split into many pieces when other ports are
added: the staircase family of \cref{thm:staircase} has two fixed ports,
no hole, and linearly growing joint cycle rank despite treewidth two.
Likewise, a treewidth bound on a sector quotient is not a treewidth
bound on the input graph---pieces can be glued through wide
interfaces---so a complete algorithm must additionally bound adhesions
and construct a bounded-width decomposition of a schedule-equivalent
reduced graph.

Nor does the cycle bound carry metric or collision information: bend
distance is not travel length, and cutting a hole open, merging
corridors, or replacing a region by shortest paths may change distances,
waiting locations, and vertex or opposite-edge conflicts.  The known
hole-free reduction requires canonicalization and a two-way,
cost-preserving schedule lift; none of those statements follows from
\(\betti(\GammaP)\le h\).

We close with the questions this paper leaves open.
\begin{enumerate}[leftmargin=2.2em]
  \item Is there a hole-aware refinement bound
  \(\operatorname{tw}(\Gamma_{\mathcal P})\le f(r,h)\) for \(r\) ports on
  pure cubical domains with \(h\) holes?  The families of
  \cref{sec:joint} constrain the possible proof shapes but are consistent
  with such a bound.
  \item What is the exact value of the edge-discretized two-port constant
  \(c_2^{\mathrm{edge}}\ge3\)?  The bridge/core argument gives the upper
  bound eleven conditionally on the metric transfer of
  \cref{rem:two-port-transfer}, and the source full version states
  seven.
  \item Does the one-sided extrusion geometry of \cref{thm:hf-osp} admit
  a controlled analogue on holed domains after the feedback set of
  \cref{cor:algorithm} is deleted?
  \item Is CMP, with total travel cost, fixed-parameter tractable in
  \(k+h\) on pure cubical domains with \(h\) holes?  We do not obtain an
  \(f(k,h)n^{O(1)}\) algorithm here.
\end{enumerate}

\appendix

\section{Machine-checked examples}
\label{app:verification}

An ancillary Python script
(\texttt{anc/verify\_examples.py} in the arXiv source package)
recomputes, from the definitions alone, the following finite assertions
of this paper:
\begin{enumerate}[leftmargin=2.2em]
  \item on the unit square of \cref{prop:endpoint}, the literal
  last-edge values \((0,0,1,2)\) by exhaustive path enumeration, the
  augmented values \((0,0,1,1)\), and the literal sector graph \(C_3\)
  of \cref{cor:literal-c3};
  \item for the staircase family of \cref{thm:staircase}, the exact
  values
  \(\betti(Q_t)=\lfloor2t/3\rfloor+2\) and
  \(B(Q_t)=\lfloor(2t+1)/3\rfloor+3\) for \(3\le t\le12\), computed by
  \(0\)-\(1\) breadth-first search in the orientation-state graph of
  \cref{sec:state} followed by component contraction;
  \item the joint labels and the refinement graph \(P_3\square P_3\) of
  \cref{prop:two-port-grid}; and
  \item exact treewidth at most three, by exhaustive elimination-order
  search, for all thirteen monotone-interface graphs of
  \cref{lem:sharp-monotone-interface}.
\end{enumerate}
These finite computations are not part of the proofs; they are included
so that the longer arithmetic certificates can be checked mechanically.

\bibliographystyle{plain}
\bibliography{references}

\end{document}